%% file: tailwind-arxiv.tex
\documentclass[sigconf, nonacm, authorversion, screen]{acmart-arxiv}
\usepackage[inline]{enumitem}
\usepackage{xcolor}
\usepackage{subcaption}
\usepackage{eso-pic}
\usepackage{amsmath, amsthm}
\usepackage{wasysym}
\usepackage{cleveref}
\usepackage{overpic}
\usepackage{tikz}
\usepackage{tabularx}
\usepackage{ragged2e}
\usepackage{algorithm}
\usepackage{algpseudocodex}

\usepackage[newfloat=true,frozencache=true]{minted}

\definecolor{ForestGreen}{RGB}{34,139,34}
\newtheorem{theorem}{Theorem}

\newcommand*\captionlabel[1]{%
  \tikz[baseline=(char.base)]{%
    \node[shape=circle,fill=white,draw=black,inner sep=1.5pt, outer sep=5pt] (char) {%
      \textcolor{black}{\scriptsize \textsf{#1}}};
}}
\newcommand*\captionlabelsmall[1]{%
  \tikz[baseline=(char.base)]{%
    \node[shape=circle,fill=white,draw=black,inner sep=1.5pt, outer sep=5pt] (char) {%
      \textcolor{black}{\tiny \textsf{#1}}};
}}

\newcommand{\sparagraph}[1]{\vspace{1mm}\noindent {\bf #1}}

\input{constants}

\begin{document}
\title{Tailwind: A Practical Framework for Query Accelerators}

\author{Geoffrey X. Yu}
\orcid{0009-0005-3186-1465}
\affiliation{%
  \institution{MIT CSAIL}
  \city{Cambridge}
  \state{MA}
  \country{USA}
}
\email{geoffxy@mit.edu}

\author{Ryan Marcus}
\orcid{0000-0002-1279-1124}
\affiliation{%
  \institution{University of Pennsylvania}
  \city{Philadelphia}
  \state{PA}
  \country{USA}
}
\email{rcmarcus@seas.upenn.edu}

\author{Tim Kraska}
\orcid{0009-0003-2414-2759}
\affiliation{%
  \institution{MIT CSAIL}
  \city{Cambridge}
  \state{MA}
  \country{USA}
}
\email{kraska@mit.edu}

\input{sections/00-abstract}

\maketitle

\input{sections/01-introduction}
\input{sections/02-user-overview}
\input{sections/03a-alp-pattern}
\input{sections/03b-alp-interface}
\input{sections/03c-alp-modeling}
\input{sections/04-details}
\input{sections/05a-eval-cases}
\input{sections/05b-eval-other}
\input{sections/06-related-work}
\input{sections/07-conclusion}

\begin{acks}
We thank Markos Markakis, Xinjing Zhou, Andreas Kipf, and Laurent Bindschaedler
for their feedback on earlier versions of this paper.
This research was supported by Amazon, Google, and Intel as part of the MIT Data
Systems and AI Lab (DSAIL) at MIT. This research was also sponsored by the
Department of the Air Force Artificial Intelligence Accelerator and was
accomplished under Cooperative Agreement Number FA8750-19-2-1000. The views and
conclusions contained in this document are those of the authors and should not
be interpreted as representing the official policies, either expressed or
implied, of the Department of the Air Force or the U.S. Government. The U.S.
Government is authorized to reproduce and distribute reprints for Government
purposes notwithstanding any copyright notation herein.
\end{acks}

\bibliographystyle{ACM-Reference-Format}
\bibliography{tailwind-arxiv}

\clearpage
\appendix
\input{sections/08-appendix}

\end{document}

%% file: constants.tex
\newcommand{\thesystem}{\textsc{Tailwind}}
\newcommand{\alp}{ALP}
\newcommand{\alps}{ALPs}
\newcommand{\tpchp}{TPC-H-Plus}

\newcommand{\OverallTpchRedshiftGeomeanSpeedup}{1.32$\times$}
\newcommand{\OverallTpchDuckDBGeomeanSpeedup}{1.38$\times$}
\newcommand{\OverallTpchDuckDBNaiveSpeedup}{0.46$\times$}

\newcommand{\OverallTpchRedshiftGeomeanSpeedupOne}{1.21$\times$}
\newcommand{\OverallTpchDuckDBGeomeanSpeedupOne}{1.33$\times$}

\newcommand{\OverallTpchpRedshiftGeomeanSpeedup}{1.37$\times$}
\newcommand{\OverallTpchpDuckDBGeomeanSpeedup}{1.76$\times$}

\newcommand{\OverallTpchpRedshiftGeomeanSpeedupOne}{1.49$\times$}
\newcommand{\OverallTpchpDuckDBGeomeanSpeedupOne}{1.78$\times$}

\newcommand{\OverallGeomeanSpeedup}{\OverallTpchDuckDBGeomeanSpeedup}
\newcommand{\WorstTpchOverallSlowdown}{\OverallTpchDuckDBNaiveSpeedup}

\newcommand{\BestGeomeanTpch}{\OverallTpchDuckDBGeomeanSpeedup}
\newcommand{\BestGeomeanTpchp}{\OverallTpchpDuckDBGeomeanSpeedup}
\newcommand{\BestGeomeanSQLStorm}{1.28$\times$}

%% file: sections/00-abstract.tex
\begin{abstract}
Relational database management systems (RDBMSes) can process general-purpose
queries, but often have lower performance compared to purpose-built solutions
for specific queries.
For example, consider a group-by query over a few known groups (e.g., grouping
by country).
While an RDBMS would likely use a hash map to do the grouping, a faster method
could hard-code the expected groups into the query executor.
Such workload-specific techniques, which we call \emph{query accelerators}, are
not widely used in practice because the engineering effort (optimizer and engine
changes, potential bugs) does not always justify the isolated performance gains
(speedup on a specific query).
We propose \thesystem{}: a non-invasive query planner that brings accelerators
into any RDBMS that supports data import/export.
Accelerator builders register accelerators using abstract logical plans
(\alps{}): a new abstraction based on regular tree expressions that specifies
the logical sub-plans each accelerator can correctly replace.
\thesystem{} also uses each \alp{}'s structure to automatically build a neural
network model to predict the accelerator's performance.
At runtime, \thesystem{} sits atop an RDBMS and transparently rewrites queries
to run across one or more accelerators when predicted to be beneficial, falling
back to the underlying RDBMS when not.
Across three distinct case studies, we use \thesystem{} to integrate
workload-specific accelerators with Redshift and DuckDB to achieve geomean
speedups of \BestGeomeanTpch{}, \BestGeomeanTpchp{}, and \BestGeomeanSQLStorm{}.
\end{abstract}

%% file: sections/01-introduction.tex
\section{Introduction}\label{sec:introduction}

\begin{figure}
  \centering
  \includegraphics[width=\columnwidth]{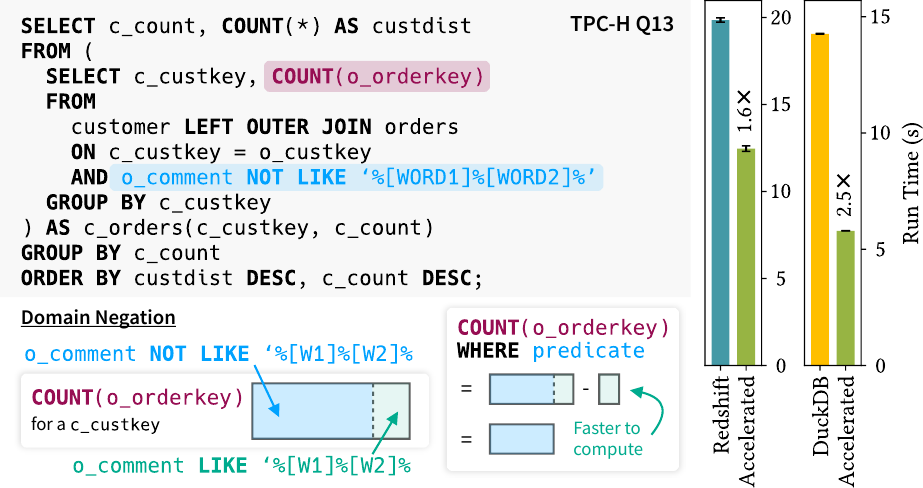}
  \caption{We can accelerate TPC-H Q13 by running the negated version of the
  inner query and \emph{subtracting} its results from the total order counts for
  each customer, netting a 1.6$\times$ and 2.5$\times$ speedup over Redshift and
  DuckDB respectively.}
  \label{fig:intro-example}
\end{figure}

Relational database management systems (RDBMSes) are widely used, in part,
because they provide an expressive and general-purpose way to query structured
data.
But this generality leaves potential performance on the table.
Purpose-built systems~\cite{scuba, nanocubes, influxdb, powerdrill, napa} are at
the other extreme: they might process a specific class of queries significantly
faster using a specialized algorithm and/or pre-computed side information, but
as a consequence, lack the generality of an RDBMS.
Consider the following example.

\sparagraph{Specialization example.}
TPC-H Q13~\cite{tpch} (Figure~\ref{fig:intro-example}),
computes a customer order count distribution for orders with comments that
do not match a pattern. An RDBMS could run this query by filtering
\texttt{orders} by the selection predicate (\texttt{NOT LIKE '\%...\%'}), hash
joining it to \texttt{customer}, and then applying two aggregations.

Although this general approach works, we can do better. The filter on
\texttt{orders} is not selective so the join involves most of the table. If we
pre-compute the total (unfiltered) order counts for each customer, we can run
the query faster using the \emph{negation} of the predicate instead; we subtract
the resulting order counts from the pre-computed totals to get the counts
matching the original predicate.\footnote{There are additional computational
details related to \texttt{NULL} that we elide for brevity and clarity in the
example.}
Since the negated predicate is selective, fewer \texttt{orders} rows match,
leading to a smaller join and thus a faster query.
Figure~\ref{fig:intro-example} compares Redshift's~\cite{aws-redshift} and
DuckDB's~\cite{duckdb} query run times to this ``domain negation'' technique,
showing a 1.6$\times$ and 2.5$\times$ speedup respectively. Although we must
pre-compute the total counts, this extra data is only 0.57\% and 0.20\% of the
total dataset size in the two systems, and it supports any predicate on
\texttt{orders}. This technique is specific: it only applies to filtered
aggregations of numeric expressions. But in exchange, it exploits the query's
structure and the underlying data to deliver a speedup.
Yet, to our knowledge, RDBMSes such as Redshift and DuckDB do not implement it.

\sparagraph{Query accelerators.}
We call such specialized techniques \emph{query accelerators}: possibly-stateful
components that can run certain logical query sub-plans faster than a target
RDBMS.
Crucially, accelerators do not have to be capable of executing every possible
query plan, giving them the freedom to exploit properties of the logical
sub-plan and underlying data for performance.

\sparagraph{Integrating accelerators.}
Accelerators help an RDBMS exploit workload-specific techniques without
sacrificing generality. Yet, integrating them into an RDBMS can be impractical,
if not impossible.
Accelerator builders cannot integrate them directly into closed-source RDBMSes
such as Redshift, and doing so in open-source systems requires substantial
engine-specific changes across the optimizer, cost model, execution engine, and
memory and storage layers.
This large ``blast radius'' makes workload-specific accelerators difficult to
upstream and maintain, and forces accelerator builders to reimplement
integration logic for each supported RDBMS.
Database extensions~\cite{dbextension-survey, pgextensions} and
UDFs~\cite{udf_outline, udo-sichert, dsos-jungmair23} seemingly offer a less
invasive alternative, but the mechanisms we know of~\cite{pgextensions,
duckdb-extensions, udo-sichert} either (i)~provide little to no query optimizer
support~\cite{graceful-udf-wehrstein25, qoff-chaudhuri93}, requiring users to
judge when an accelerator will help and explicitly rewrite their queries to use
them; or (ii)~require code changes to the underlying
engine~\cite{dsos-jungmair23}.
Extensions and UDFs also couple accelerators to a specific RDBMS, making them
hard to reuse elsewhere.
What is missing is a portable, non-invasive approach that automatically inserts
accelerators into query plans when expected to help.

\sparagraph{\thesystem{}.}
We introduce \thesystem{}: a non-invasive query planner and executor for adding
query accelerators to any RDBMS that supports data import and export.
Accelerator builders declare the logical plan fragments their accelerator can
compute and provide an implementation against \thesystem{}'s execution API.
Offline, given a representative query log~\cite{redset,
cloud_analytics_benchmark} and space budget, \thesystem{} learns the
accelerators' performance, selects accelerator instances to minimize predicted
query run time, and manages their side information (if any).
Online, when a query arrives, \thesystem{} acts like a database
proxy~\cite{rds-proxy, pgbouncer} and transparently rewrites queries to use
those accelerator instances when beneficial.
\thesystem{} orchestrates execution across the accelerators and the underlying
RDBMS, transferring intermediate data as needed so that accelerators can replace
sub-plans anywhere within a query plan.
\thesystem{} is complementary to techniques that synthesize specialized query
executors~\cite{castor-feser, bespoke-olap, gendb, gpt-db}, focusing instead on
the problem of intelligently deciding when and where to use accelerators in a
query.
The end result is a system that automatically and \emph{surgically} inserts
accelerators into an RDBMS without modifying its internals, helping to close the
gap between specialized and general-purpose systems.

\sparagraph{Challenges.}
While simple to use, realizing \thesystem{} presents a significant design
challenge. We need a generic way to (\textbf{C1})~let accelerator builders
describe which logical plan fragments their accelerator can correctly replace,
(\textbf{C2})~expose query-specific information needed by the accelerator
implementation, and (\textbf{C3})~automatically model the accelerator's
performance to know when to use it. For example, for domain negation, builders
need a simple but precise way to describe the ``filtered aggregation'' pattern,
the implementation needs to know what predicate to negate and what aggregates to
pre-compute, and \thesystem{} needs to learn how these pieces influence the
accelerator's performance.

\sparagraph{Key component: Abstract logical plans (\alps{}).}
To check if an accelerator can be used on a query (\textbf{C1}), a na\"ive
approach is to try defining its applicability using a view and rely on view
matching algorithms~\cite{mvopt, mvopt-survey}. But views are insufficient
because they express a single concrete logical plan whereas an accelerator might
support a variable-sized sub-plan (e.g., a left-deep join chain).
On the other hand, requiring builders to write all of the code that checks if
their accelerator is usable on a given query is also undesirable. This is
because custom code provides no common structure for extracting an accelerator's
parameters (e.g., the predicate to negate) (\textbf{C2}) and is hard to
featurize for performance modeling (\textbf{C3}).

To address these challenges, we introduce \emph{abstract logical plans}
(\alps{}): a new abstraction that accelerator builders use to describe their
accelerator's semantics.
At its core, an \alp{} centers on a parameterized regular tree
expression~\cite{tata, rte-impl} over logical query plans that specifies the set
of logical plan fragments an accelerator can correctly replace (\textbf{C1}).
The parameters exist to capture query-specific information that can vary across
accelerator uses (e.g., the predicate to negate in domain negation).
The same \alp{} serves two further purposes. It can be compiled into a tree
automaton~\cite{tata, tree-acceptors-doner70, tree-automata-thatcher68} to
efficiently search query plans for matches and extract parameter values
(\textbf{C2}), and its tree structure and parameters can be used to construct a
custom neural network that estimates when the accelerator is beneficial to use
(\textbf{C3}, Section~\ref{sec:key-ideas-model}).
For \thesystem{}, \alps{} therefore provide a unified abstraction for an
accelerator's applicability, interface to its implementation, and performance
modeling.

We demonstrate the generality of \thesystem{} and \alps{} using three case
studies on distinct query workloads. Over all our cases, \thesystem{}
lets us integrate workload-specific accelerators to speed up queries on Redshift
and DuckDB, achieving geomean query run time speedups of \BestGeomeanTpch{},
\BestGeomeanTpchp{}, and \BestGeomeanSQLStorm{}.

\ \\[-0.25em]
\noindent
\textbf{Contributions.}
In summary, we make the following contributions:
\begin{itemize}[leftmargin=*]
  \item Abstract logical plans (ALPs): an abstraction that describes the logical
  plan fragments an accelerator supports, interfaces with the implementation,
  and enables performance modeling.
  \item A specialized graph neural network model featurization for ALPs, used to
  learn performance models for each accelerator.
  \item The \thesystem{} system, which allows for automatically and surgically
  inserting builder-defined accelerators into query plans.
  \item The implementation and evaluation of these ideas in \thesystem{} over
  three case studies.
\end{itemize}

%% file: sections/02-user-overview.tex
\section{\thesystem{} Overview}

\begin{figure*}[t]
  \centering
  \begin{overpic}[width=0.98\textwidth]{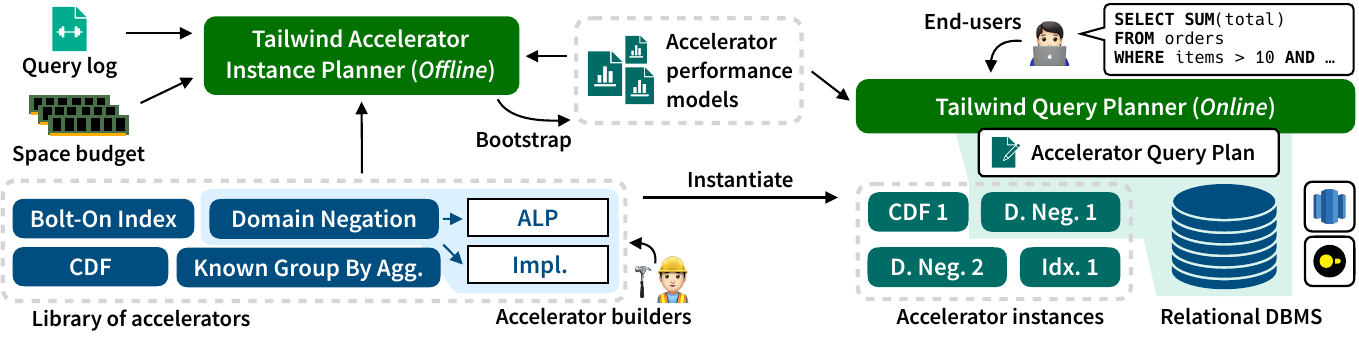}
    \put(0,0.85){\captionlabel{A}} % Library
    \put(34.75,15.85){\captionlabel{B}} % Bootstrap models
    \put(1,23){\captionlabel{C}} % Workload
    \put(-0.5,17){\captionlabel{D}} % Space Budget
    \put(16,21.5){\captionlabel{E}} % Offline
    \put(63.5,0.85){\captionlabel{F}} % Instances
    \put(64,16.75){\captionlabel{G}} % Online
    \put(94.5,13){\captionlabel{H}} % Tailwind plan
  \end{overpic}
  \caption{An overview of the \thesystem{} system and its workflow.}
  \label{fig:overview}
\end{figure*}

We begin with an overview of how accelerator builders and end-users use
\thesystem{} (\Cref{fig:overview}).
We see accelerator builders as data engineers with expertise in query processing
and workload specialization and end-users as any user of the underlying RDBMS;
the roles can overlap too.
Accelerator builders register a ``library'' of accelerators~\captionlabel{A},
each comprising an \alp{} (\Cref{sec:key-ideas-alp}) and implementation.
Offline, \thesystem{} then learns a run time model~\captionlabel{B} for each
accelerator (\Cref{sec:key-ideas-model}) by generating accelerator
instances and measuring their run times to get training data
(\Cref{sec:bootstrapping-details}).

Since accelerators can be stateful (e.g., domain negation pre-computes the
unfiltered aggregations), \thesystem{} next chooses the accelerator
\emph{instances} to create.
End-users provide a query log,~\captionlabel{C} which is a representative
(not exhaustive) list of previously seen queries, and a space
budget~\captionlabel{D}.
\thesystem{}'s offline planner~\captionlabel{E} takes these inputs and finds
accelerator instance candidates by enumerating \alp{} matches in the log.
It then uses its performance models in a greedy search to choose
instances~\captionlabel{F} that optimize the workload's predicted speedup
subject to the space budget (\Cref{sec:key-ideas-inst}).

Online, end-users submit their queries to \thesystem{}'s query
planner~\captionlabel{G}. Using its instantiated accelerators,~\captionlabel{F}
\thesystem{} rewrites queries to use one or more accelerator instances if they
apply and are predicted to be beneficial~\captionlabel{H}.
It then coordinates query execution by invoking the accelerators and underlying
RDBMS, moving intermediate data to/from these components as needed.

%% file: sections/03a-alp-pattern.tex
\section{Abstract Logical Plans}
Recall that abstract logical plans (\alps{}) serve three key roles in
\thesystem{}:
\begin{enumerate*}[label=(\roman*)]
  \item describing the logical plan fragments an accelerator can replace,
  \item exposing matched plan elements and parameters to the accelerator
    implementation, and
  \item providing structure for performance modeling.
\end{enumerate*}
We look at each role in detail next.

\subsection{Specifying Accelerator Applicability}\label{sec:key-ideas-alp}
\noindent
\textbf{A na\"ive non-solution.}
A tempting idea is to express an accelerator's applicability using a database
view and rely on view matching~\cite{mvopt, mvopt-survey} to check if the
accelerator can be used on a query.
However, views fall short because they only define a single concrete query
expression bound to specific relations.
Since an accelerator could potentially handle a variable-sized sub-plan (e.g.,
any left-deep join of arbitrary depth), we need a solution that is capable of
specifying a (possibly infinite sized) set of logical query sub-plans.

\sparagraph{Overview.}
To specify the set of logical plan fragments that an accelerator can correctly
replace, an \alp{} comprises:
\begin{enumerate}[leftmargin=*]
  \item A \textbf{\emph{template}}, which is a parameterized regular tree
  expression~\cite{tata, rte-impl} over logical query plans that describes the
  ``shape'' of the query plan fragments the accelerator supports.
  \item A \textbf{\emph{validator}}, which is a builder-defined function that
  takes a match instance (the parameter values extracted from a query fragment
  that matched the template) and returns true if the accelerator supports the
  match instance.
\end{enumerate}
An accelerator can be used on a query if the template matches a logical sub-plan
in the query \emph{and} the validator returns true for the match instance.
\alps{} need both components because templates alone cannot express non-regular
patterns nor patterns that have semantic constraints (e.g., that two tables
being joined are a primary to foreign key join).
Using only a validator (i.e., only builder-written matching code) is undesirable
because arbitrary code provides no common structure for extracting parameters
and is hard to featurize for performance modeling
(\Cref{sec:alp-model-details}).
Informally, the template captures the ``structure'' of the pattern and a
validator ``refines'' it as needed.
We now describe these two components in detail.

\subsubsection{Templates}\label{sec:alp-templates}
To help explain templates, we first provide some background on regular tree
expressions (RTEs)~\cite{tata}.

\sparagraph{Background on RTEs.}
RTEs generalize regular expressions~\cite{theorycomp} to tree expressions. RTEs
are themselves a directed tree. Using function notation to describe tree
expressions (e.g., $t(a, b)$ denotes a tree of three nodes with root $t$, left
child $a$, and right child $b$), an RTE is defined inductively as~\cite{tata}:
\begin{itemize}[leftmargin=*]
  \item A \emph{leaf token} $a$. The RTE matches a tree with a single $a$ node.
  \item A \emph{tree token} $t(R_1, \dots, R_n)$. The RTE matches a tree with
    root node $t$ and $n$ children where child $i$ matches $R_i$.
  \item An \emph{alternation} $R_1 \mid R_2$, which matches RTEs $R_1$ or $R_2$.
  \item A \emph{concatenation} $R_1 \odot_{c} R_2$, which means an RTE composed
    of $R_1$ with its leaf node $c$ replaced with $R_2$.
  \item A \emph{Kleene star}: $R^{c*}$, which matches $R$ zero or more times
  where the root of each repetition is at the leaf node $c$ (appears in $R$).
\end{itemize}
For example, the RTE $R = \left[t(c, b)\right]^{c*} \odot_{c} a$ would match
$t(a, b)$; $t(t(a, b), b)$; $t(t(t(a, b), b), b)$; and so on.

\sparagraph{From RTEs to Templates.}
While RTEs provide a formal foundation for tree pattern matching, they are not
usable out-of-the-box in \thesystem{} for two key reasons.
First, classical RTEs lack a mechanism for indicating the parts of a pattern
that will vary across distinct matches (e.g., the predicate to negate in domain
negation).
Second, recall that \thesystem{} has an offline phase where it selects
accelerators to instantiate and an online phase where it searches for matches of
these instantiated accelerators within incoming queries.
We need a way to distinguish between parts of the pattern that must be resolved
and fixed at instantiation-time (e.g., the specific aggregates to compute in
domain negation) versus the parts resolved online in an incoming query (e.g.,
the predicate to negate).
These requirements motivate our design of \alp{} templates.

\begin{table}[t]
  \centering
  \caption{The variable types used in \alp{} templates.}
  \label{tab:alps-var-types}
  \footnotesize
  \begin{tabularx}{\columnwidth}{@{}lX@{}}
    \toprule
    \textbf{Variable type} & \textbf{Description} \\
    \midrule
    Column reference & A reference to a column (e.g., in a predicate or aggregation) \\
    Table reference & A reference to a table (e.g., in a scan) \\
    \midrule
    Table expression & A query sub-plan (e.g., scan, join, filter) \\
    Column expression & An expression involving columns (e.g., col1 + col2) \\
    Boolean expression & An expression that produces a boolean (e.g., a predicate) \\ 
    \bottomrule
  \end{tabularx}
\end{table}

\sparagraph{\alp{} Templates.}
An \alp{} template $T$ is also defined inductively and includes leaf
tokens and tree tokens like classical RTEs. They include three more constructs
to support \thesystem{}'s use case:
\begin{itemize}[leftmargin=*]
  \item \textbf{Typed variables.} A typed variable $\mathbf{Var}(x;
  \tau)$ has an identifier $x$ and matches any plan expression of type $\tau$.
    Variables are the parameters in the template and represent parts of the
    pattern that will vary across distinct matches.
    We use types to introduce matching constraints (e.g., the matched variable
    being a column reference versus a column expression) and to provide semantic
    information to featurize for performance modeling
    (Section~\ref{sec:alp-model-details}).
    Table~\ref{tab:alps-var-types} lists the variable types \alps{} support.
  \item \textbf{Alternations.} An alternation $\mathbf{Alt}(x; T_1, \dots, T_n)$
  has identifier $x$ and is otherwise the same as an RTE's alternation.
  \item \textbf{Repetitions.} A repetition $\mathbf{Rep}(x; T_B, T_R, c)$ has an
  identifier $x$ and represents a template $T_R$ that repeats zero (or one) or
  more times at the $c$ leaf node, followed by one occurrence of
    $T_B$.
    Repetitions are implemented using an RTE's Kleene star and concatenation.
    Instead of exposing those constructs, \alp{} templates use repetitions to
    serve a dual practical purpose: (i) providing a single constrained way to
    describe arbitrarily-deep but bounded repetition in query sub-plans (e.g.,
    nested joins), which in turn provides (ii) a single construct representing
    repeating plan substructures to featurize for performance modeling
    (Section~\ref{sec:alp-model-details}).
\end{itemize}
If the same variable (i.e., having the same identifier) appears multiple times
in the \alp{}, the template will only match query sub-plans where each
occurrence of the variable resolves to the same value.
Repetitions are an exception: each repeating instance creates a new scope.
Variables and alternations inside a repetition can resolve to different values
in different repeating instances.
Finally, for each variable, alternation, and repetition, accelerator builders
must indicate if the construct is resolved at \emph{instantiation time} (i.e.,
during \thesystem{}'s offline planning) or \emph{online} (i.e., when matching
accelerator instances to queries just before execution).

\subsubsection{Validator}
Validators have the following signature:
\begin{minted}{rust}
fn is_valid(match_inst, schema) -> bool
\end{minted}
Validators are called with a \texttt{match\_inst} (described in
\Cref{sec:alp-interface}) and the database \texttt{schema} (used if checking
semantic constraints, e.g., that a join is primary to foreign key).

\sparagraph{Expressiveness.}
\alps{} can describe any decidable logical query sub-plan pattern. This is
because the template can match an arbitrary sub-plan (e.g., using a table
expression variable), and the validator can be any decidable function.
However for practical ease-of-use and performance modeling, what
matters is how much of a pattern's structure can be expressed using the
declarative \alp{} template. Empirically, in our case studies
(Section~\ref{sec:eval-case-study}), we were always able to express the
pattern's core structure using the template and only relied on the validator to
check semantic constraints (e.g., that a matched join condition is a
primary-foreign key join).

\subsubsection{\alp{} Example}\label{sec:alp-template-example}
To provide intuition for \alps{}, we go through an \alp{} template for the
domain negation accelerator. \alp{} templates are directed trees. Therefore, for
clarity, we visualize the template in Figure~\ref{fig:dneg-alp} instead of
relying on notation.
We start with a table expression variable~\captionlabel{A} matching any logical
query plan. We declare that it is resolved at instantiation-time because the
pre-computed aggregates depend on this sub-plan.
Next, we declare a filter token~\captionlabel{B} since the template must match a
filter over some sub-plan. The filter has a online-resolved predicate
variable~\captionlabel{C} because the accelerator supports any predicate over
its fixed sub-plan.
Finally, we add a group by aggregation token~\captionlabel{D} since the matched
sub-plan must be an aggregation.
We add a repetition~\captionlabel{E} over a new instantiation-time resolved
column reference variable~\captionlabel{F} since there can be multiple key
columns in the group by.
Similarly, we create another repetition~\captionlabel{G} over a new
instantiation-time resolved variable~\captionlabel{H} since there can be
multiple aggregation expressions.
The accelerator supports \texttt{SUM} and \texttt{COUNT}, so we wrap the
aggregation expression variable in an alternation with these two
options~\captionlabel{J}.

\begin{figure}[t]
  \begin{overpic}[width=\columnwidth]{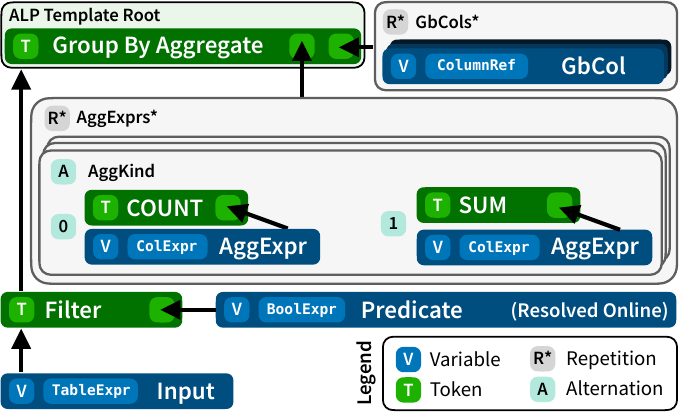}
    \put(36,2.5){\captionlabel{A}} % Input variable
    \put(16.5,14.5){\captionlabel{B}} % Filter
    \put(69.5,14.5){\captionlabel{C}} % Predicate
    \put(26,58){\captionlabel{D}} % Group by aggregate
    \put(72,57){\captionlabel{E}} % Gbcol repetition
    \put(93.5,50.5){\captionlabel{F}} % Gbcol variable
    \put(27,42.7){\captionlabel{G}} % Agg repetition
    \put(49,23.5){\captionlabel{H}} % Agg expr variable
    \put(27,35){\captionlabel{J}} % Agg expr alternation
  \end{overpic}
  \caption{The domain negation \alp{} template (tree depiction).}
  \label{fig:dneg-alp}
\end{figure}

%% file: sections/03b-alp-interface.tex
\subsection{Interfacing with Implementations}\label{sec:alp-interface}
When an \alp{} matches a query, we say that it produces a \emph{match instance},
which is a map that contains
\begin{itemize}[leftmargin=*]
  \item The matched value for each variable in the \alp{}'s template
  \item The matched option for each alternation in the template
  \item The number of times each repetition matched, with the variable
    values and alternation options for each repeating instance
\end{itemize}
An \alp{} and match instance define a concrete logical plan that the accelerator
can execute. A match instance thus represents the ``query-specific information''
that an accelerator implementation needs during \thesystem{}'s offline planning
and online at runtime.

\sparagraph{Accelerator API.}
To add an accelerator to \thesystem{}, accelerator builders need to implement
three functions:
\begin{itemize}[leftmargin=*]
  \item \begin{minted}{rust}
fn alp() -> ALP
\end{minted}
  \item \begin{minted}{rust}
fn build(match_inst, ctx) -> State
\end{minted}
  \item \begin{minted}{rust}
fn run(state, match_inst, input, ctx) -> Output
\end{minted}
\end{itemize}
The \texttt{alp()} function returns the accelerator's \alp{}.
Given a match instance and context about the underlying RDBMS' SQL semantics,
\texttt{build()} performs any pre-computation needed to create an accelerator
instance (e.g., computing the unfiltered aggregates in domain
negation).
\thesystem{} invokes \texttt{run()} with the match instance, the state returned
by \texttt{build()}, any input data, and the same context just mentioned.
Builders must ensure that their accelerator implementation returns the same
results that the underlying RDBMS would for the logical sub-plan that they
replace.
Accelerators at the bottom of a query plan only produce outputs. Accelerators in
the middle of a query plan receive inputs and produce outputs.
\thesystem{} uses Arrow~\cite{arrow} as the input and output data format.

%% file: sections/03c-alp-modeling.tex
\subsection{Accelerator Performance Modeling}\label{sec:key-ideas-model}
Finally, we look at how \alps{} support performance modeling.
A key challenge is that \thesystem{} must support builder-defined accelerators
without making assumptions about how they are implemented, which precludes prior
approaches that featurize UDF code~\cite{graceful-udf-wehrstein25}.

We address this challenge by observing that \alps{} provide a structured
description of what an accelerator computes. Critically, the template
distinguishes between the parameterized and constant parts of the accelerator
and encodes their dependencies. For example, in the domain negation template
(Figure~\ref{fig:dneg-alp}), the input query plan varies across accelerator
instantiations and its value influences the accelerator run time, which the
template's tree structure expresses.

\thesystem{} thus models an accelerator's performance by encoding its \alp{}
template as a tree-structured neural network. The learned weights correspond to
the varying parts of the template (its variables, alternations, and repetitions)
and the message passing follows the dependency relationships in the \alp{}
(\Cref{sec:alp-model-details}).
\thesystem{} learns one model per accelerator. We use the same featurization
strategy for each, but each model will have a different set (and number) of
weights because they are based on different \alp{} templates.

\sparagraph{Bootstrapping.}
Before \thesystem{} uses its accelerators, it must train its models offline.
This only needs to be done once per accelerator. \thesystem{} uses the \alp{}
definitions to generate accelerator instances to collect training data; we
detail this process in \Cref{sec:bootstrapping-details}.

\section{Query Planning With Accelerators}\label{sec:query-planning-idea}
Bringing everything together, we now look at \thesystem{}'s planners.

\sparagraph{Offline planning.}\label{sec:key-ideas-inst}
Before processing any queries, \thesystem{} selects accelerator candidates to
\emph{instantiate}, given a user-provided space budget in bytes $S$. We do this
because accelerators are parameterized and can be stateful; typically the space
usage of all possible candidates greatly exceeds $S$. This offline step ensures
\thesystem{}'s online planner only considers the ``most useful'' set of
accelerator options given $S$.
This selection problem is a generalized version of automatic index selection,
which itself is NP-hard~\cite{indexselhard-chaudhuri04, indexselopt-comer78,
indexselopt-piatetsky83}.

\thesystem{} tackles this challenge using a two step approach. First, it
enumerates a set of possible accelerator instance candidates by matching
its \alps{} against the user-provided query workload to narrow the candidate
search space. Then, it uses a greedy search over the instance candidates,
picking candidates that provide the best predicted run time reduction normalized
by its space usage.
We found this strategy to be acceptable in selecting useful candidates in
practice (Section~\ref{sec:eval-case-study}) and discuss details in
Section~\ref{sec:search-details}.

\sparagraph{\alp{} matching.}
To find accelerator candidates, \thesystem{} searches for \alp{} matches in a
query plan.
There can be multiple equivalent versions of a query plan and only some may
match an \alp{}.
\thesystem{} handles this by converting the query plan into an
e-graph~\cite{egraphs-book, egraphs-thesis-nelson} and applying equality
saturation~\cite{eqsat-tate09}. E-graphs are compact data structures that encode
equivalent versions of expression trees (e.g., logical query plans), similar to
the tree of groups in the Cascades optimizer~\cite{cascades-graefe95}.
\thesystem{} uses e-graphs for engineering reasons~\cite{egg-willsey21}; our
techniques would also apply to a Cascades optimizer.
\thesystem{} searches for \alp{} matches by converting the \alp{} template into
a nondeterministic finite tree automata~\cite{tata} and checking for matches in
each e-class within the e-graph. We describe our matching and \alp{} resolving
algorithms in Section~\ref{sec:matching-details}.

\sparagraph{Online planning.}
At runtime, \thesystem{} receives a query and searches for \alp{} matches using
the instantiated candidates' \alps{}. It enumerates all possible ways to execute
the query using the instances that match and selects the option predicted
to be fastest.
An exhaustive search is practical because, at runtime, few accelerator instances
match a query (zero to two in our experimental workload).
We evaluate this planner's overhead in Section~\ref{sec:eval-details}.

\sparagraph{Query performance modeling.}
Given a query with accelerator(s), \thesystem{} individually predicts the run
times of (i) the accelerators, (ii) any data transfer that would occur, and
(iii) the remaining query.
It sums these predictions to estimate an end-to-end run time.
For~(i), it uses the learned model for the accelerator
(Section~\ref{sec:key-ideas-model}).
For~(ii), it uses a linear model based on the amount of data transferred
(using a cardinality estimate).
For~(iii), \thesystem{} scales the run time of the bare query (without
accelerators) by the ratios of the logical plan operators and estimated input
cardinality between the bare and remaining queries.
\thesystem{} assumes access to (i) a standard cardinality estimator
(e.g.,~\cite{selingeropt-selinger79}) whose statistics are periodically
maintained, similar to how production RDBMSes routinely run \texttt{ANALYZE},
and (ii)
a query run time predictor (e.g.,~\cite{qppnet-marcus19, stage-wu24, orig_pred,
learning_latency, opt_est_par, opt_est}, orthogonal to this work); see
Section~\ref{sec:other-models-details}.

%% file: sections/04-details.tex
\section{Additional \thesystem{} Details}

\subsection{\alp{} Model Details}\label{sec:alp-model-details}
Recall that \thesystem{} models the performance of an accelerator by
representing its \alp{} template as a tree-structured neural network. We now
look at how \thesystem{} constructs and trains these models.

\begin{figure}[t]
  \begin{overpic}[width=\columnwidth]{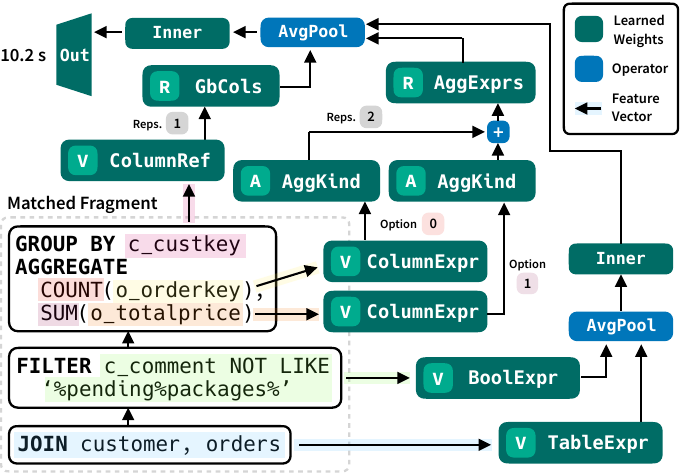}
    \put(42,34){\captionlabel{A}} % Plan fragment
    \put(95,36){\captionlabel{B}} % Internal
    \put(34,68){\captionlabel{B}} % Internal
    \put(76,35){\captionlabel{C}} % Alternation
    \put(76,49){\captionlabel{D}} % Repetition
    \put(35,50){\captionlabel{D}} % Repetition
    \put(2,65){\captionlabel{E}} % Final output
  \end{overpic}
  \caption{The domain negation accelerator's featurization.}
  \label{fig:alp-model-example}
\end{figure}

\subsubsection{Featurization}
We encode the \alp{} template's parameterized constructs (i.e., variables,
repetitions, and alternations) and then propagate and pool the encoded values up
the template tree following its directed edges. At the \alp{}'s root, we pass
the output vector through a multi-layer perceptron (MLP)~\cite{dlbook-2016} to
get a predicted run time.
We use this model architecture because it captures the varying parts of an
accelerator (e.g., what predicate to negate, which aggregates to compute) and
encodes their dependencies using its tree structure.
We walk through an example featurization for the domain negation accelerator
(Figure~\ref{fig:alp-model-example}) to illustrate the details.

\begin{table}[t]
  \caption{The features used for each variable type.}
  \label{tbl:variable-features}
  \footnotesize
  \begin{tabularx}{\columnwidth}{@{}lX@{}}
    \toprule
    \textbf{Variable type} & \textbf{Features} \\
    \midrule
    Table expression & Cardinality, width (bytes), Stage / T3 op.
      features~\cite{stage-wu24, t3-rieger25} \\
    Table reference & Cardinality, width (bytes) \\
    \midrule
    Boolean expression & Selected fraction (\%) \\
    Column exp. or ref. & Cardinality, number of unique values \\
    \bottomrule
  \end{tabularx}
\end{table}

\sparagraph{Learned weights.}
We create one encoder MLP per variable type, repetition, and alternation in the
\alp{} template. We also create one MLP for inner template nodes and one output
MLP. Our domain negation \alp{} template (Figure~\ref{fig:dneg-alp}) has four
variable types, one alternation, and two repetitions. So we create nine MLPs:
seven for the parameterized constructs plus the two additional ones. The teal
boxes in Figure~\ref{fig:alp-model-example} depict the learned MLPs. The boxes
with the same name refer to the same MLP; we show duplicates only to make the
model's inference procedure easier to understand.

\subsubsection{Inference}
Suppose our domain negation \alp{} matches the plan fragment shown in
Figure~\ref{fig:alp-model-example}~\captionlabel{A}. \thesystem{} creates a run
time prediction by traversing the \alp{} template and recursively constructing
an output embedding vector bottom-up by using the following procedure for each
\alp{} construct.

\sparagraph{Variables.}
Variables are always leaf nodes in the \alp{} template. We create a vector
representation for each variable's value, using the set of features shown in
Table~\ref{tbl:variable-features}. We then pass each variable vector through its
corresponding encoder MLP to obtain an output vector.

\sparagraph{Internal tokens.}
We take the output vectors from the token's children and perform average pooling
to obtain a single fixed-size vector. We then pass this vector through the inner
template node MLP to obtain an output vector~\captionlabel{B}.

\sparagraph{Alternations.}
We take the output vector from the option that matched in the template
instance. We concatenate a one hot feature vector corresponding to the
chosen option and pass this new vector through the alternation's MLP to obtain
an output vector.

\sparagraph{Repetitions.}
We take the output vectors from the repetition's children and add them together
elementwise. We concatenate the number of repetitions to the summed vector and
pass this new vector through the repetition's MLP to get an output
vector~\captionlabel{D}.

\sparagraph{Final output.}
We take the output vector from the \alp{}'s root and pass that through
the output MLP to get a run time~\captionlabel{E} prediction.

\subsubsection{Bootstrapping Details}\label{sec:bootstrapping-details}
To learn a model for each accelerator, \thesystem{} requires labeled performance
data. Although \alps{} are for recognizing accelerator patterns, we can
re-purpose them to \emph{generate} possible accelerator instances. \thesystem{}
measures these instances' run times to get labeled data for training.

The high level idea is to select random values for each construct
(e.g., variable) in the \alp{} and return the instances that pass the
validator.
However, this random selection must ensure semantic correctness and diversity.
For example, in the domain negation \alp{} (Figure~\ref{fig:dneg-alp}), the
sampled predicate variable value must involve columns accessible from the input
plan and \thesystem{} needs to avoid generating only true (or only false)
predicates.
To generate instances, \thesystem{} first selects repetition counts and
alternation options, which fixes the ``structure'' of the resulting instance
(e.g., selecting the number of group by columns in domain negation).
Then, it traverses the \alp{} template bottom up and computes the set of
``available columns'' at that stage in the query fragment (e.g., the output
schema of the child query plan node) and then selects values for any variables
using the strategies summarized in Table~\ref{tbl:generation-method}, taking
care to only select available columns.
\thesystem{} repeats this procedure to generate a sufficient number of training
data instances per accelerator.
We found that 10k instances were sufficient, but this number is a hyperparameter
that users can adjust if needed.

\begin{table}[t]
  \caption{The methods we use to generate variable values.}
  \label{tbl:generation-method}
  \footnotesize
  \begin{tabularx}{\columnwidth}{@{}lX@{}}
    \toprule
    \textbf{Variable type} & \textbf{We randomly select} \\
    \midrule
    Table reference & A table in the dataset \\
    Table expression & A query slice from the planning workload \\
    \midrule
    Column reference & A column in the available columns \\
    Column expression & A column, binary arithmetic expression (column and
    literal, or two random columns), or depth 2 binary expression \\
    Boolean expression & An open range, a closed range, equality, two
      predicates joined using \texttt{AND} or \texttt{OR} \\
    \bottomrule
  \end{tabularx}
\end{table}

\subsection{\thesystem{}'s Other Models}\label{sec:other-models-details}
To help decide whether using accelerators is better than the base RDBMS,
\thesystem{} also leverages other models. They are not the focus of this paper,
but we describe them here for completeness.

\sparagraph{Query run time.}
\thesystem{} needs to estimate the run time of a query. Since this is a problem
that has been extensively studied~\cite{t3-rieger25, stage-wu24, lce-sun19,
qppnet-marcus19, zeroshot-hilprecht22, unify-wu22, neo-marcus19, bao-marcus22,
brad-yu24}, we assume an acceptable model exists. In our evaluation, we assume
we have the query's run time on the RDBMS (e.g., it is cached) and study
\thesystem{}'s robustness to query run time prediction errors
(Section~\ref{sec:eval-robustness}).
For remaining queries (the operators after parts of a query have been replaced
by accelerator(s)), \thesystem{} scales the bare query's run time by the ratio
of logical operators and estimated input cardinalities between the bare and
remaining plans. This coarse approach is effective because accelerators typically
replace large portions of a query.

\sparagraph{Data transfer.}
We estimate \thesystem{}'s intermediate data transfer time using
$\mathrm{max}(S/k_i + S/k_e, C)$ where $S$ is the size of the data transferred
(e.g., in bytes) and $k_i$, $k_e$, and $C$ are empirically measured constants
representing the data import and export rates and $C$ is a minimum transfer
time.
\thesystem{} measures these constants offline before handling user queries.
We use S3~\cite{s3} to transfer data between the accelerators and Redshift. For
DuckDB, accelerators run in-process, so there is a negligible transfer overhead.

\subsection{\alp{} Matching and Resolution}\label{sec:matching-details}
\thesystem{}'s planner searches for accelerator use opportunities in a workload.
This process consists of three conceptual parts: (i) matching \alps{} to
queries, (ii) resolving the matched constructs in the \alp{} (e.g., assigning
concrete values to its variables), and (iii) enumerating query plans that
contain the matched accelerators.

\subsubsection{\alp{} Matching}
\alps{} match patterns in tree expressions after being compiled into
non-deterministic finite tree automata (NFTAs)~\cite{rtenfta-belabbaci18,
tree-acceptors-doner70}. NFTAs are analogous to NFAs for
strings~\cite{theorycomp} (see \Cref{apdx:matching} for background).
However, simply matching \alps{} to query plans is insufficient because the
query might need to be rewritten before it matches (e.g., reordering a filter in
the query plan).
\thesystem{} addresses this problem by converting queries into an
e-graph~\cite{egg-willsey21}, using equality saturation to find equivalent
queries~\cite{eqsat-tate09}, and then matching \alps{} on the resulting e-graph.

\sparagraph{NFTA Matching Algorithm on E-Graphs.}
E-graphs are data structures that represent equivalent tree expressions (e.g.,
multiple equivalent logical query plans) in a single data
structure~\cite{egraphs-thesis-nelson, egraphs-book}, similar to logical groups
in a Cascades optimizer~\cite{cascades-graefe95}.
\thesystem{} uses a top-down NFTA matching algorithm~\cite{tata} on each e-class
in the query's e-graph. Since existing NFTA matching algorithms are for tree
expressions (not e-graphs), we make two key extensions.

Our first extension is to search over all e-nodes in the e-class. We accept
(match) the root e-class if there is at least one e-node in the root class where
all of its children reach a leaf transition. We examine all e-nodes in the
e-class because the NFTA might only match a tree rooted at one of the nodes in
the class (i.e., corresponding to a specific way of expressing the accelerator's
plan fragment).

Our second extension is in rejecting match cycles. NFTA transitions can contain
state cycles and e-graphs can contain cycles. Match algorithms for trees do not
have this problem because tree expressions are acyclic, thus match cycles are
impossible. Our approach is to reject matching paths that contain \emph{both} a
transition cycle and e-class cycle. Rejecting only transition cycles is
incorrect because they describe repetitions in the tree expression. Rejecting
only e-class cycles is incorrect because the NFTA could be matching a tree
expression that has a finite number of recursive repetitions.
See Algorithm~\ref{alg:nfta-match} in \Cref{apdx:matching} for complete details.

\subsubsection{Resolving ALP Constructs}
Next, \thesystem{} resolves (i.e., assigns values to) the \alp{}'s variables,
alternations, and repetitions to provide the necessary inputs for its models and
the accelerator's implementation.
When building an \alp{}'s NFTA~\cite{rtenfta-belabbaci18}, \thesystem{} stores
the identifier of the ``accepting state'' associated with each template
construct.
During matching, it records the e-classes it encounters when following a
transition from these stored states, which lets it map an \alp{} construct to
its matched e-class in the e-graph.
\thesystem{} resolves a variable by extracting the smallest tree expression from
its matched e-class. It resolves alternations by checking which option's output
state matched.
For repetitions, which can contain nested variables and alternations,
\thesystem{} records discovery and finishing logical timestamps during e-graph
matching.
It can then unambiguously map nested variables and alternations to their correct
repetition instance by finding the smallest interval that contains the matched
construct's traversal interval.

\subsubsection{Extracting Accelerator Plans}
Finally, \thesystem{} enumerates accelerator execution plans using the e-graph.
During matching, \thesystem{} inserts a new e-node for the accelerator into the
e-class that it matches. This indicates that the accelerator is equivalent to
the other expressions in the e-class.
Then, to enumerate the query execution options, \thesystem{} constructs the
options recursively per e-class.
For each non-accelerator e-node, we compute its execution options as the
cartesian product of all options for its children e-classes. There are typically
few accelerator candidates that match a query (zero to two times), so the
cartesian product is small in practice and usually has a size of one. We union
each node's options with the accelerator(s) to get the set of options for that
e-class. \thesystem{} then returns the execution options from the root e-class.

\begin{figure*}[t]
  \centering
  \begin{subfigure}[t]{0.47\textwidth}
    \begin{overpic}[width=\textwidth]{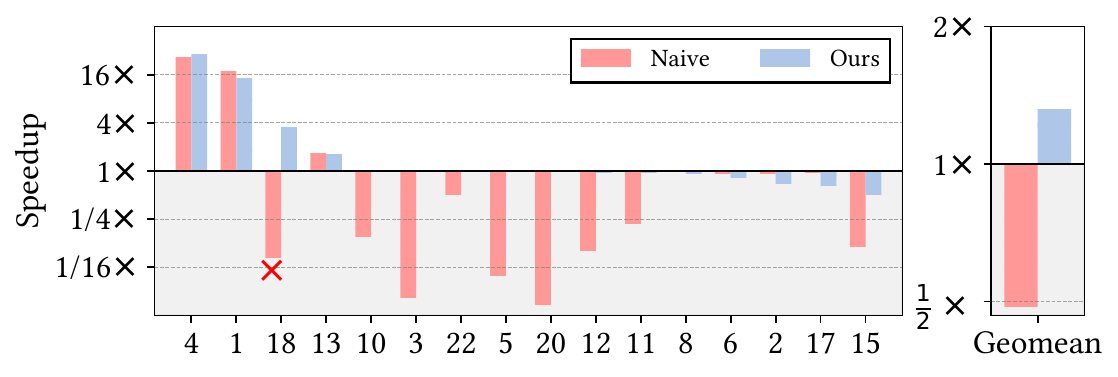}
      \put(24,25){\captionlabelsmall{A}}
    \end{overpic}
    \caption{Redshift (10\%) (X indicates timeout)}
    \label{fig:tpch-redshift10}
  \end{subfigure}
  \qquad
  \begin{subfigure}[t]{0.47\textwidth}
    \begin{overpic}[width=\textwidth]{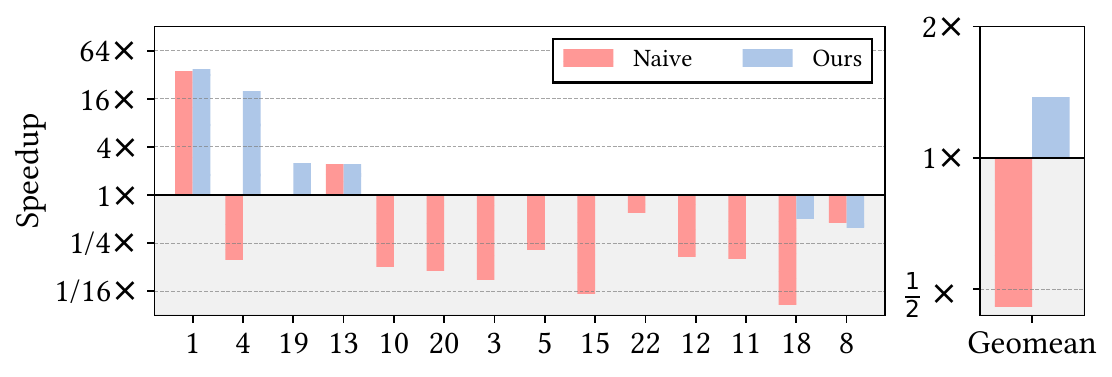}
      \put(20.5,27.5){\captionlabelsmall{B}}
      \put(25,21.5){\captionlabelsmall{C}}
    \end{overpic}
    \caption{DuckDB (10\%)}
    \label{fig:tpch-duckdb10}
  \end{subfigure}
  \caption{\thesystem{} accelerates TPC-H (SF = 100) on both Redshift and DuckDB
  by \OverallTpchRedshiftGeomeanSpeedup{} and \OverallTpchDuckDBGeomeanSpeedup{}
  respectively (geomean).}
  \label{fig:tpch-e2e}
\end{figure*}

\subsection{Selecting Accelerator Instances}\label{sec:search-details}
Recall that \thesystem{} uses a greedy search to select the accelerator
instances to use for a storage budget. \thesystem{} first enumerates possible
instance candidates by matching its \alps{} against the user-provided query
log.
It then evaluates these candidates by computing their normalized benefit: the
predicted reduction in workload execution time divided by the accelerator's
space usage. The algorithm iteratively selects the instance with the highest
normalized benefit, recomputing benefits after each selection, until the there
is no more run time reduction or it exhausts the space budget.

\sparagraph{Discussion.}
We take this two step approach for two reasons.
First, matching \alps{} to queries in the query log to enumerate candidates
ensures we only consider instances that would match the workload. This strategy
prunes the search space using our assumption that workload log is representative
of the online workload.
Second, we use a greedy search because our candidates usually affect one query
each and multiple candidates have an monotonic effect on a query's speedup.
While complex interactions are possible (e.g., two candidates accelerating a
query together but not individually), our experiments show that a greedy search
is practical (\Cref{sec:eval-case-study}).

\sparagraph{Analysis.}
Let $n$ be the number of accelerator instance candidates. Our greedy search
has a complexity of $O(n^2)$ since it recomputes the normalized benefit (for the
remaining candidates) each time it selects a candidate. An exhaustive search
would consider $O(2^n)$ candidate sets. The algorithm terminates because, on
each iteration, it selects a candidate and the set of candidates is finite.

%% file: sections/05a-eval-cases.tex
\begin{figure*}
  \centering
  \begin{subfigure}[t]{0.47\textwidth}
    \begin{overpic}[width=\textwidth]{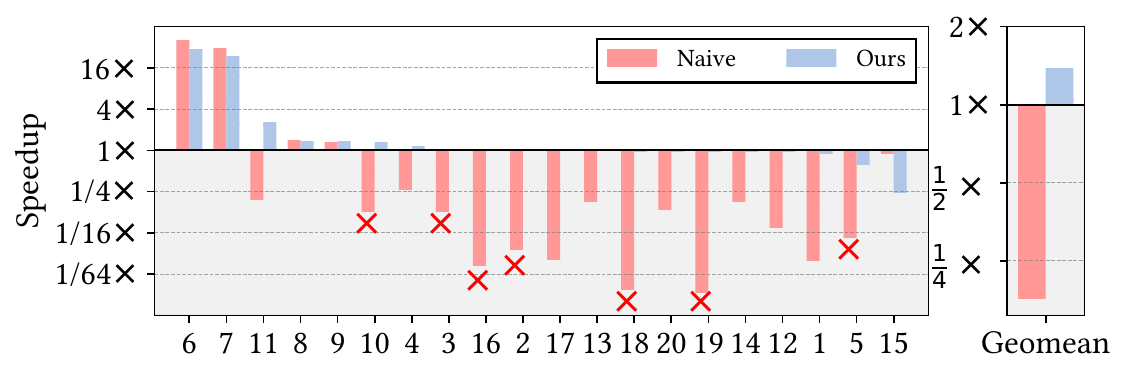}
      \put(31,9.5){\captionlabelsmall{E}}
    \end{overpic}
    \caption{Redshift (10\%) (X indicates timeout)}
    \label{fig:tpchp-redshift10}
  \end{subfigure}
  \qquad
  \begin{subfigure}[t]{0.47\textwidth}
    \begin{overpic}[width=\textwidth]{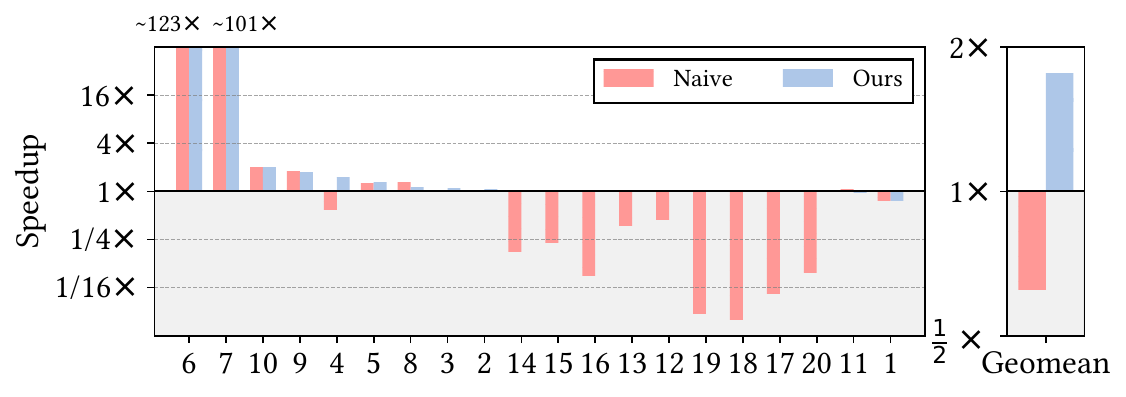}
      \put(16.5,14.5){\captionlabelsmall{D}}
    \end{overpic}
    \caption{DuckDB (10\%)}
    \label{fig:tpchp-duckdb10}
  \end{subfigure}
  \caption{\thesystem{} accelerates \tpchp{} on Redshift and DuckDB
  by \OverallTpchpRedshiftGeomeanSpeedup{} and \OverallTpchpDuckDBGeomeanSpeedup{}
  respectively (geomean).}
  \label{fig:tpchplus-e2e}
\end{figure*}

\section{Workload Case Studies}\label{sec:eval-case-study}
Across three distinct query workloads, we now show how users can use
\thesystem{} to incorporate workload‑specific accelerators that speed up
selected queries on both closed and open source RDBMSes.
We begin by describing our common experimental setup.

\sparagraph{Base systems.}
We run \thesystem{} with Redshift~\cite{aws-redshift} and DuckDB~\cite{duckdb}.
Redshift is a closed-source managed distributed database system. DuckDB is an
embedded single-process database system.
We run Redshift on a cluster of two \texttt{ra3.large} nodes. We run DuckDB
v1.3.2 on an \texttt{r6id.xlarge} EC2 instance~\cite{ec2}, placing the DuckDB
file on its NVMe drive. The \texttt{r6id.xlarge} instance has the same number of
vCPUs and memory as the Redshift cluster in aggregate.

\sparagraph{Accelerator compute.}
For DuckDB, we run the accelerators on the same \texttt{r6id.xlarge} machine as
the RDBMS. As Redshift is a managed system, we run the accelerators on our own
machine (20 core Intel Xeon Gold 6230, 128 GiB of memory). We limit the
accelerators to 4 threads so they use similar compute resources.

\sparagraph{Na\"ive baseline.}
We compare \thesystem{} to a na\"ive strategy that always uses an accelerator if
it matches the query. It greedily selects the instances that replace the most of
a query until it exhausts the space budget. At runtime, if there are multiple
ways to use the accelerators, it again picks the option that replaces the most
of the query with accelerators.
We use a timeout of 5 minutes to avoid ``getting stuck'' running poor plans
chosen by the na\"ive strategy.

\sparagraph{Implementation.}
We implemented \thesystem{} and the accelerators used in these case studies in
Rust using $\sim$65,000 lines of code.

\subsection{Case Study 1: TPC-H}\label{sec:eval-cs1}
We start by looking at TPC-H~\cite{tpch}, which is a standard analytical query
benchmark. We use scale factor 100 (i.e., 100 GB of data).

\subsubsection{Accelerator Library}\label{sec:eval-cs1-library}
The first step in using \thesystem{} on a workload is to identify accelerators
that exploit the characteristics of the queries and/or underlying dataset.
Recall that the goal of this work is not to propose new general optimizations
that apply to all queries, but to design a new framework that integrates
workload-specific accelerators with an RDBMS.
We use four such accelerators in this case study (not necessarily an exhaustive
set of accelerators).

\sparagraph{Domain negation.}
This is the accelerator described in Section~\ref{sec:introduction}.

\sparagraph{Cumulative filtered aggregates (CDFs).}
Another acceleration opportunity is in queries involving
invertible aggregations~\cite{incrementalagg-tangwongsan15} (\texttt{SUM},
\texttt{COUNT}, \texttt{AVG}) over data filtered by an ordered column (e.g., a
date range).
If we pre-compute a cumulative aggregation over the range, answering these query
fragments reduces to looking up the cumulative values for the range endpoints
and returning the difference, which is faster than executing the query from
scratch.

\sparagraph{Known group by with fused arithmetic.}
We can accelerate queries that aggregate over a few known groups by hard-coding
the expected groups into a query executor and fusing the aggregation computation
into one compute kernel.
For example, Query 1 from TPC-H only produces four groups and performs
arithmetically-heavy aggregations~\cite{compiled-vect-queries-kersten18,
tpchanalysis-boncz13}.
The performance gain is from (i) avoiding a hash map for grouping and (ii)
keeping the intermediate aggregations in registers during execution.
This technique differs from standard query compilation, as hash maps are used to
collect groups~\cite{querycompilation-neumann11}. While a sophisticated query
compiler could in theory use this specialization too, general-purpose systems
would not.

\sparagraph{Ordered indexes.}
Some analytical RDBMSes (e.g., Redshift) do not have
indexes~\cite{redshift-docs, bigquery-docs}. Indexes (i.e., a secondary sort
order) could help speed up selective scans and point lookups.

\subsubsection{Results and Key Takeaways}
We run \thesystem{} with a 10\% and 1\% space budget (i.e., a 10~GiB and 1~GiB
budget for a 100~GiB dataset), which is comparable in size to other
techniques that use additional space to accelerate analytical
queries~\cite{copyright, sagedb, parachute}.
We generate one query instance per query template in the workloads;
\thesystem{}'s offline planner uses these queries to select accelerator
instances. We then run the workload with a different set of query instances
(i.e., with different placeholder values).
Note that we exclude query 21 from the experiment because one of its subqueries
(that \thesystem{} considers) triggers a query planning performance bug in
Redshift (the subquery takes over an hour to \texttt{EXPLAIN}).

Figure~\ref{fig:tpch-e2e} shows our results on Redshift and DuckDB with a 10\%
space budget; see \Cref{apdx:tpch1} for the 1\% plots. The $y$-axis is the
speedup relative to running the full query on the RDBMS and is in log scale;
higher is better. Bars below 1$\times$ represent a slowdown. The $x$-axis is the
query sorted by decreasing speedup.
We report a speedup summary by taking the geomean across all queries in the
workload~\cite{tpch}. For clarity, we only plot the queries where there was a
speedup (or slowdown); for some queries, no accelerator was used so their
performance is unchanged.

\textbf{\thesystem{} speeds up queries on both RDBMSes by
\OverallGeomeanSpeedup{} on average (geomean).}
On TPC-H with a 10\% (1\%) space budget, \thesystem{} accelerates queries by
\OverallTpchRedshiftGeomeanSpeedup{} (\OverallTpchRedshiftGeomeanSpeedupOne{})
on Redshift and \OverallTpchDuckDBGeomeanSpeedup{}
(\OverallTpchDuckDBGeomeanSpeedupOne{}) on DuckDB (geomean).
The speedup on TPC-H comes from the CDF (Q1, Q4), index (Q18, Q19), and domain
negation (Q13) accelerators. The known group by accelerator also applies to Q1,
but \thesystem{} selects the CDF as it provides a larger speedup.
This result shows that \thesystem{} can leverage its accelerators across a
diverse set of queries, RDBMSes, and space budgets.

While \thesystem{} achieves overall speedups, there are a few queries where it
adds a small slowdown. In all but four cases, the slowdown is from
\thesystem{}'s optimization time, as those queries are short running (hundreds of
milliseconds). We study \thesystem{}'s overhead in
Section~\ref{sec:eval-overhead}.
For DuckDB Q18, the slowdown is due to a misprediction
in the remaining query's run time. While the possibility of mispredictions is
fundamental, \thesystem{} can cache its mistakes and avoid the poor plan the
next time the query arrives. This mitigation is effective in repetitive query
workloads, which are common in practice~\cite{redset}.

\textbf{Na\"ively using accelerators when they match leads to significant
slowdowns, as slow as \WorstTpchOverallSlowdown{} (DuckDB).}
The slowdowns on Redshift are mostly due to data transfer: moving intermediate
results from the accelerator to Redshift (or vice-versa) can take longer than
running the entire query on Redshift.
DuckDB has a negligible data transfer overhead as the accelerators run
in-process. Its slowdowns are because the na\"ive strategy uses the domain
negation accelerator whenever it matches, even when not beneficial. When the
predicate is already selective, its negation processes more data (and is thus
slower). Thus overall, considering an accelerator's full effect on the query is
essential for achieving workload speedups using accelerators.

\textbf{When multiple accelerator options exist for a query, \thesystem{}'s
performance models help it select a better option.}
For Q18 and Q4 in Figures~\ref{fig:tpch-redshift10}~\captionlabel{A}
\ref{fig:tpch-duckdb10}~\captionlabel{B}, \thesystem{} and the na\"ive strategy
make different decisions. For Q18, the na\"ive strategy uses domain negation
because it can push down most of the query into the accelerator. But that
decision leads to a slowdown. In contrast, \thesystem{} uses an index on a
subquery of Q18, which it correctly predicts as the faster option.
Q4 is similar: the na\"ive strategy uses domain negation but \thesystem{}
selects the CDF accelerator as it correctly predicts it to be the faster
option.

\textbf{\thesystem{} successfully tailors its accelerator choices to the
underlying RDBMS.}
For TPC-H Q19~\captionlabel{C}, \thesystem{} uses an index on DuckDB but avoids
it on Redshift as it correctly predicts that the data transfer would be too slow
on Redshift.

\begin{figure}[t]
  \centering
  \includegraphics[width=\columnwidth, clip, trim={0cm 0.22in 0cm 0cm}]{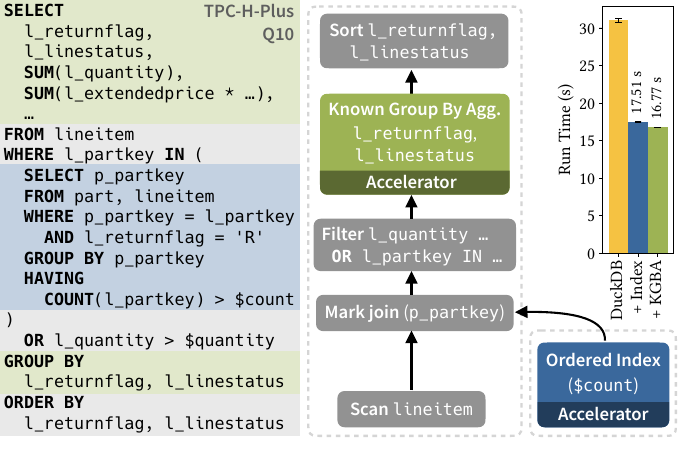}
  \caption{\thesystem{} leverages two accelerators to speed up \tpchp{} Q10 by
  1.9$\times$ on DuckDB.}
  \label{fig:two-accel-example}
\end{figure}

\subsection{Case Study 2: \tpchp{}}\label{sec:eval-cs2}
As TPC-H queries are structurally simple, we develop an extended workload called
\tpchp{} for this next case study. \tpchp{} contains 20 query templates that
merge fragments of TPC-H queries (hence its name) to provide opportunities for
multiple accelerator matches per query. Half of the workload's queries match one
or more accelerators that provide a speedup. The other half are queries that
match our accelerators but would slow down the query if used. We run \tpchp{}
over a scale factor 100 TPC-H dataset.

\subsubsection{Accelerator Library}
Since \tpchp{} is based on TPC-H query fragments, we use the same accelerator
library (Section~\ref{sec:eval-cs1-library}).

\subsubsection{Results and Key Takeaways}
Like our first case study (Section~\ref{sec:eval-cs1}), we run \thesystem{} with
a 10\% and 1\% space budget and use the same methodology to create offline
planning queries and online queries for execution.
Figure~\ref{fig:tpchplus-e2e} shows results for a 10\% space budget
(\Cref{apdx:tpchp1} has the 1\% budget plots), from which we draw the following
conclusions.

\textbf{\thesystem{} speeds up \tpchp{} queries on Redshift by
\OverallTpchpRedshiftGeomeanSpeedup{} and DuckDB by
\OverallTpchpDuckDBGeomeanSpeedup{} on average (geomean).}
With a 10\% (1\%) space budget, \thesystem{} accelerates queries by
\OverallTpchpRedshiftGeomeanSpeedup{} (\OverallTpchpRedshiftGeomeanSpeedupOne{})
on Redshift and \OverallTpchpDuckDBGeomeanSpeedup{}
(\OverallTpchpDuckDBGeomeanSpeedupOne{}) on DuckDB (geomean).
The speedup comes from all four accelerators. Q6 and Q7 see large speedups
because they use the CDF accelerator, which leverages pre-computation. These
queries have a higher speedup on DuckDB (Figure~\ref{fig:tpchp-duckdb10}
\captionlabel{D}) because it has a negligible data transfer cost.
This result shows that \thesystem{} identifies opportunities to use its
accelerators across a diverse set of queries, RDBMSes, and space budgets.

While \thesystem{} achieves overall speedups, it introduces a slowdown on
Redshift Q5, Q15, and DuckDB Q1 due to a misprediction in the accelerator's run
time or the remaining query's run time. While possible mispredictions are
fundamental, \thesystem{} can cache its mistakes and avoid the poor plan the
next time the query arrives (effective in repetitive query workloads, common in
practice~\cite{redset}).

\textbf{Like in our first case study, \thesystem{} (i) avoids accelerator
instances that introduce slowdowns and (ii) tailors its accelerator choices to
the underlying RDBMS.}
Many of the na\"ive strategy's timeouts in Figure~\ref{fig:tpchp-redshift10} are
because it blindly applies the known group by accelerator. Known group by is
usually a poor choice for Redshift, as \thesystem{} must transfer the input data
from the RDBMS to the accelerator to aggregate (usually large). The slowdowns in
Figure~\ref{fig:tpchp-duckdb10} are due to the domain negation accelerator,
which is not always a good choice to use. \thesystem{}, in contrast, leverages
its models to avoid these poor choices.

\textbf{When beneficial, \thesystem{} correctly leverages multiple
accelerators per query.}
It achieves geomean speedups of 5.62$\times$ on Redshift (Queries 7, 9) and
4.60$\times$ on DuckDB (Queries 5, 7, 9, 10) (averaged among these queries).
Figure~\ref{fig:two-accel-example} shows an example of how \thesystem{} uses the
known group by accelerator together with the index over a subquery to speed up
query 10 on DuckDB.
Tailwind uses multiple accelerators more frequently on DuckDB because it can
exploit the known group by accelerator; on Redshift, high data transfer
overheads make this accelerator a poor choice.
For example, the nai\"ve strategy slows down query 10 on Redshift because it
incorrectly applies the known group by (Figure~\ref{fig:tpchp-redshift10}
\captionlabel{E}).

\begin{figure}
  \centering
  \begin{overpic}[width=\columnwidth]{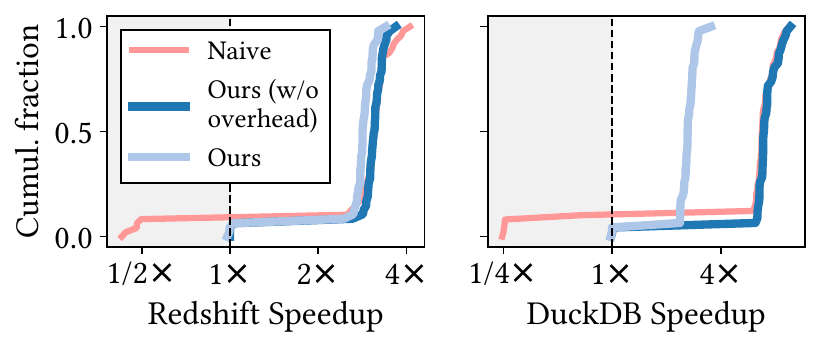}
    \put(86.5,30){\captionlabel{G}} % Overhead
    \put(20,12){\captionlabel{H}}
    \put(65,12){\captionlabel{H}}
  \end{overpic}
  \caption{Speedups on our SQLStorm Stack Overflow queries.}
  \label{fig:eval-cs3}
\end{figure}

\subsection{Case Study 3: SQLStorm Stack Overflow}
In our final case study, we work with the Stack Overflow
dataset~\cite{stackoverflow-dataset} and SQLStorm queries~\cite{sqlstorm25}.
Unlike TPC-H, which uses synthetic data, the Stack Overflow dataset comprises
240 GB of real world data from the Stack Overflow
website~\cite{stackoverflow-dataset}. SQLStorm is a query workload comprising
$\sim$18k diverse queries generated for this schema~\cite{sqlstorm25}.

\subsubsection{Accelerator Library}
This case study involves different queries and data, so we add two accelerators
that are well-suited to SQLStorm. These are not necessarily an exhaustive set
for SQLStorm.

\sparagraph{Top-k pushdown.}
A common SQLStorm query pattern involves joining tables followed by a top-$k$
operation (\texttt{ORDER BY} with a small \texttt{LIMIT}). A known optimization
is to push the top-$k$ operation below the join~\cite{topkpush-carey97}, which
reduces the number of tuples processed by the join. But this optimization is
only correct if the join does not eliminate any of the pre-selected top-$k$
rows.
We empirically found that neither Redshift nor DuckDB performs this pushdown,
even when we explicitly declare foreign key constraints.
Thus, we leverage domain-specific knowledge of the Stack Overflow schema to
identify join cases where this pushdown can be correctly applied, and add these
schema-specific cases to a stateless accelerator.

\sparagraph{Badges MV.}
SQLStorm queries include subqueries aggregating user badge counts, which can be
accelerated using a classical materialized view (MV) in \thesystem{}. We
implement this MV with a ``sparse table'' optimization that decouples users with
zero badges from badge earners. Over half of Stack Overflow users have no
badges~\cite{stackoverflow-dataset}, so the accelerator only stores their IDs
instead of full rows for zero counts. Compared to classic MVs, this approach
lowers the accelerator's space usage for skewed distributions.

\sparagraph{Previous accelerators.}
The domain negation, CDF, and ordered index accelerators
(Section~\ref{sec:eval-cs1-library}) also match queries in this workload, so we
include them in \thesystem{}'s accelerator library as well.

\subsubsection{Workload Setup}
SQLStorm contains around 18k distinct queries~\cite{sqlstorm25}. To construct a
workload, we sample 1000 queries. We then execute them on Redshift and DuckDB
and retain the queries that complete without an error under 60 seconds on both
engines, resulting in 356 queries. We use a \texttt{r6id.2xlarge} instance for
DuckDB to have a large enough NVMe drive for the dataset.

\subsubsection{Results and Key Takeaways}
Similar to our previous case studies, we use a 10~GB space budget, which
corresponds to $\sim$4.5\% of the total dataset size. We use half of the sampled
queries (178) in \thesystem{}'s offline planning pass and run other half online.
Figure~\ref{fig:eval-cs3} shows our results, plotted as a speedup CDF. The
$x$-axis is the achieved speedup (in a log scale) over running the queries
directly on the RDBMS.
We plot the na\"ive strategy, \thesystem{}, and \thesystem{} without its query
optimization overhead. For clarity, Figure~\ref{fig:eval-cs3} only plots the
queries where an accelerator matched.

\textbf{Averaging over the queries in our workload, \thesystem{} speeds up
queries by 1.28$\times$ and 1.27$\times$ (geomean) on Redshift and DuckDB
respectively.}
On both RDBMSes, the speedup mostly comes from the top-$k$ accelerator. The
difference between \thesystem{} and the na\"ive strategy at the upper end
(Figure~\ref{fig:eval-cs3}~\captionlabel{G}) is due to \thesystem{}'s
optimization overhead (performance prediction).
This overhead is similar on both systems, but is a larger proportion of the
query run times on DuckDB (the queries run faster). Note that the query run
times are not directly comparable between Redshift and DuckDB because they are
running on different hardware.

\textbf{Like previous case studies, \thesystem{}'s models help it avoid
accelerators that slow down a query.}
While the na\"ive approach can occasionally luck into the right accelerator
choice, i.e., gaining the benefits of \thesystem{} without the optimization
overheads, it will also make poor choices that lead to regressions. On both
Redshift and DuckDB, it incorrectly uses domain negation and introduces
slowdowns~\captionlabel{H}.
\thesystem{} pays some optimization overhead, therefore reducing our relative
advantage on short running queries, but it has the benefit of avoiding such
regressions.

%% file: sections/05b-eval-other.tex
\section{Drill Down Evaluation}\label{sec:eval-details}
We now examine how \thesystem{}'s components contribute to its
performance. We seek to answer the following questions:
\begin{itemize}[leftmargin=*]
  \item How accurate are \thesystem{}'s accelerator performance models and how
  do different models affect its decisions? (Section~\ref{sec:eval-learning})
  \item How does the space budget affect speedups?
  (Section~\ref{sec:eval-budget})
  \item How do query run time prediction errors affect \thesystem{}'s query
  acceleration decisions? (Section~\ref{sec:eval-robustness})
  \item What is \thesystem{}'s run time overhead? (Section~\ref{sec:eval-overhead})
  \item How do writes impact \thesystem{}'s acceleration?
  (Section~\ref{sec:eval-writes})
\end{itemize}

\noindent
We run our drill down evaluation on TPC-H and \tpchp{} as they sufficiently
capture the representative behavior of \thesystem{}.

\begin{table}[t]
\caption{The \alp{} neural network's prediction error (p50, p90 Q-error)
compared to simpler models (lower is better).}
\label{tbl:alp-modeling}
\small
\begingroup
\setlength{\tabcolsep}{4pt}
\begin{tabularx}{\columnwidth}{@{}X rr rr rr rr@{}}
  \toprule
   & \multicolumn{2}{c}{\textbf{Mean val.}} & \multicolumn{2}{c}{\textbf{Reg. forest}} & \multicolumn{2}{c}{\textbf{MLP}} & \multicolumn{2}{c@{}}{\textbf{ALP NN}} \\
   \cmidrule(lr){2-3} \cmidrule(lr){4-5} \cmidrule(lr){6-7} \cmidrule(l){8-9}
   \textbf{Accelerator} & \textbf{p50} & \textbf{p90} & \textbf{p50} & \textbf{p90} & \textbf{p50} & \textbf{p90} & \textbf{p50} & \textbf{p90} \\
   \midrule
   D. Neg. (Redshift) &  3.58 & 11.50 & 1.99 & 8.50 & 2.25 & 6.96 & \textbf{1.43} & \textbf{3.36} \\
   D. Neg. (DuckDB) &  3.02 & 1196 & 1.76 & 10.58 & 2.18 & 8.20 & \textbf{1.31} & \textbf{2.74} \\
   CDF & 11.05 & 131.4 & \textbf{1.31} & 2.81 & 1.36 & 3.45 & 1.34 & \textbf{2.15} \\
   Ordered Index & 3.29 & 16.34 & \textbf{1.39} & 2.67 & 1.46 & 2.57 & 1.48 & \textbf{2.53} \\
   Know Gr. By Agg. & 61.36 & 89.85 & \textbf{1.06} & 1.21 & 1.09 & 1.23 & 1.07 & \textbf{1.18} \\
  \bottomrule
\end{tabularx}
\endgroup
\end{table}

\begin{figure}[t]
  \centering
  \includegraphics[width=\columnwidth]{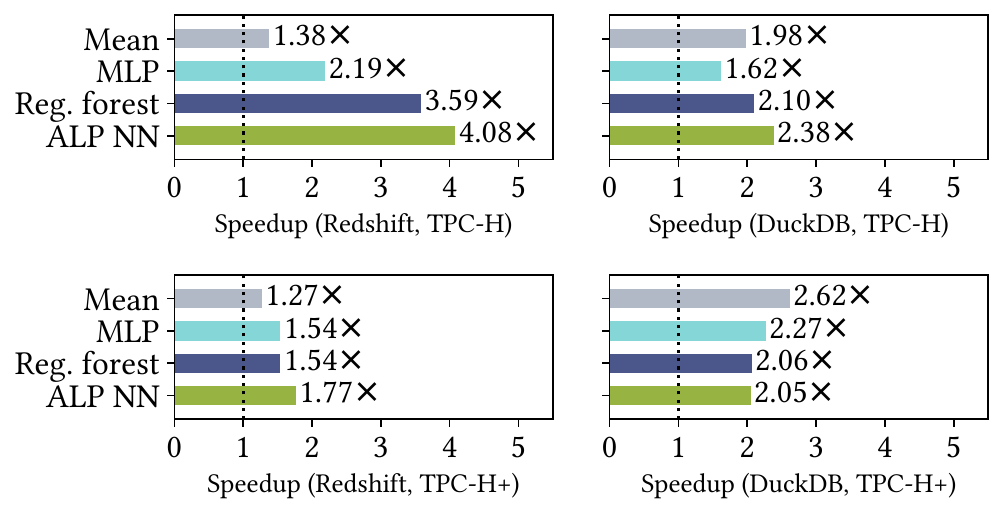}
  \caption{End-to-end effect of different accelerator modeling strategies on
    query performance (higher speedup is better).}
  \label{fig:model-e2e-effect}
\end{figure}

\subsection{Accelerator Performance Models}
\label{sec:eval-learning}

\subsubsection{Prediction Accuracy.}
We train an \alp{} neural network for the accelerators in our first case study
(Section~\ref{sec:eval-cs1-library}) using a dataset of $\sim$10k data points
per accelerator with an 80/10/10 train, validation, test split.
We compare our \alp{} neural network to three baselines:
\begin{enumerate*}[label=(\roman*)]
  \item always predicting the mean (geomean) run time,
  \item a regression random forest~\cite{randomforest-ho95} (3 trees), and
  \item a multi-layer perceptron (MLP)~\cite{dlbook-2016} (3 layers).
\end{enumerate*}
For the regression forest and MLP, we manually design a feature vector for each
accelerator.
We evaluate prediction accuracy using Q-error, defined as $\text{max}\left(p/a,
a/p\right)$, where $p$ and $a$ are the predicted and actual run times
respectively~\cite{qerror-moerkotte09}. Lower is better and 1.0 is the best
possible Q-error.
Table~\ref{tbl:alp-modeling} lists our results, from which we draw the
following conclusions.

\textbf{The \alp{} neural network achieves the lowest p90 Q-error on
our tested accelerators and a p50 Q-error comparable to the best baseline.}
On the domain negation accelerator, our \alp{} neural network achieves the
lowest Q-error compared to the other baselines. This is because (i) it is harder
to model compared to the other accelerators in our library, and (ii) our \alp{}
neural network provides a better inductive bias (encoding the accelerator's
logical structure) compared to a simple MLP and regression forest.

\textbf{Our \alp{} model featurization is general enough to work across diverse
accelerators.}
Unlike our baselines, which use hand-designed features, our \alp{} model uses
one featurization method for all the accelerators
(Section~\ref{sec:alp-model-details}) and achieves the best p90 Q-error.

\begin{figure}[t]
  \includegraphics[width=\columnwidth]{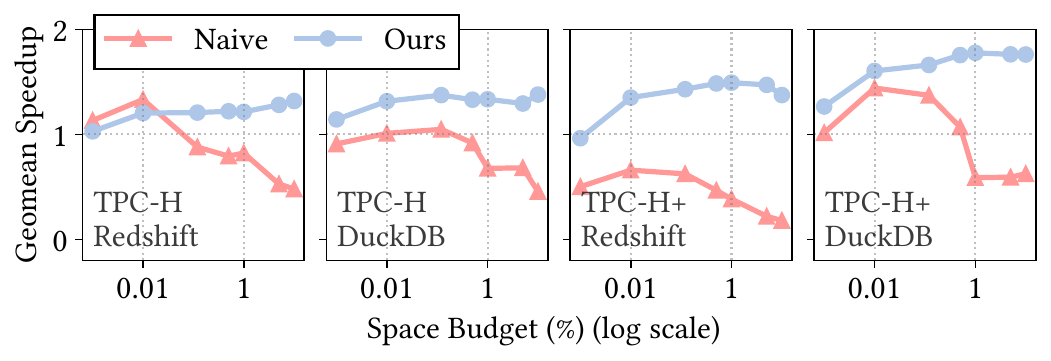}
  \caption{Workload speedup for varied space budgets.}
  \label{fig:e2e-space-sweep}
\end{figure}

\subsubsection{End-to-End Effect.}
Next, we study how different accelerator performance models affect
\thesystem{}'s query execution decisions. We run our workloads with each model
option and compare the geomean achieved speedup.
Figure~\ref{fig:model-e2e-effect} shows our results; higher is better.
Here we compute the geomean only among the queries where \thesystem{} uses an
accelerator to highlight the impact of different models. From these results, we
draw the following conclusions.

\textbf{Non-trivial learned models achieve the highest speedups and avoid
slowdowns.}
Except for \tpchp{} on DuckDB, predicting the mean value does not lead to the
highest speedup. On Redshift, using the mean value causes slowdowns as
\thesystem{} mistakenly chooses an accelerator that is slower than the base
system.
For \tpchp{} on DuckDB, all model options happen to result in the same
accelerator choices. The mean value and MLP baselines have higher speedups
because they require fewer features (meaning a faster optimization time)
compared to the other two models.

\textbf{The \alp{} NN provides the best speedup in three of the four
database/workload combinations.}
It makes a better run time prediction for a domain negation accelerator
instance, which lets \thesystem{} correctly choose that instance in the three
cases (whereas the other models do not). While the simple models do well, they
use hand-designed features for each accelerator. In contrast, the \alp{} neural
network model uses one common featurization method.

\subsection{Space Budget Sensitivity}\label{sec:eval-budget}
Next, we study the space budget's effect on \thesystem{}'s acceleration as it
varies from 0.001\% to 10\%. Figure~\ref{fig:e2e-space-sweep} shows our results.
The horizontal axis is in log scale and the vertical axis is the workload's
geomean speedup (higher is better). We draw two conclusions.

\textbf{\thesystem{}'s speedup increases with a larger space budget and it
accelerates the workload with budgets as small as 0.01\%.}
This is because we can use more accelerator instances given a larger budget. At
small space budgets (up to 0.01\%), the na\"ive strategy's speedup increases;
there are only a few usable candidates and they happen to lead to a speedup.
However, cost models are still important to use, as the na\"ive strategy causes
slowdowns at large space budgets and its speedup is always below 1.0$\times$ for
\tpchp{} on Redshift. The na\"ive strategy has a slightly higher speedup
compared to \thesystem{} on TPC-H Redshift up to a budget of 0.01\% due to the
computational overhead of using \thesystem{}'s models.

\textbf{With larger space budgets, the na\"ive strategy degrades
performance.}
A larger space budget allows for more accelerator candidates, but the na\"ive
strategy does not consider if these options are beneficial before using them
(hence the worsening slowdowns). On Redshift, the slowdowns are mostly from slow
intermediate result transfers. On DuckDB, the slowdowns beyond a 1\% budget are
because the na\"ive strategy selects slow domain negation instances.

\begin{figure}[t]
  \includegraphics[width=\columnwidth]{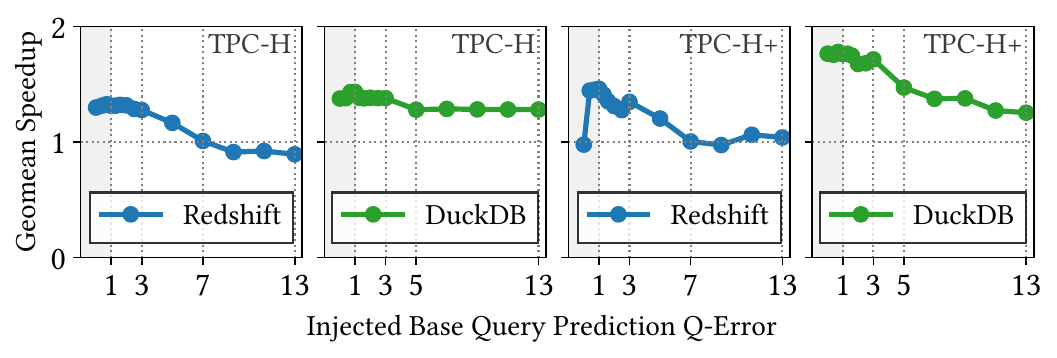}
  \caption{Robustness to query run time prediction error.}
  \label{fig:qpp-robustness}
\end{figure}

\begin{figure*}[t]
  \centering
  \begin{subfigure}{0.23\textwidth}
    \centering
    \includegraphics[width=\textwidth]{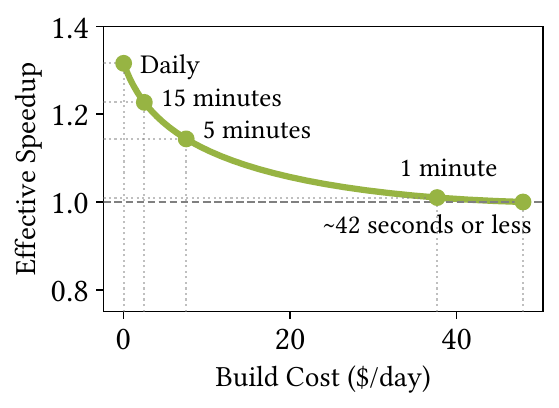}
    \caption{Redshift (TPC-H)}
  \end{subfigure}
  \hfill
  \begin{subfigure}{0.23\textwidth}
    \includegraphics[width=\textwidth]{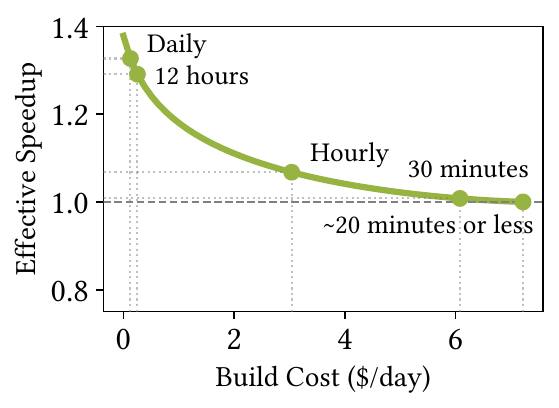}
    \caption{DuckDB (TPC-H)}
  \end{subfigure}
  \hfill
  \begin{subfigure}{0.23\textwidth}
    \centering
    \includegraphics[width=\textwidth]{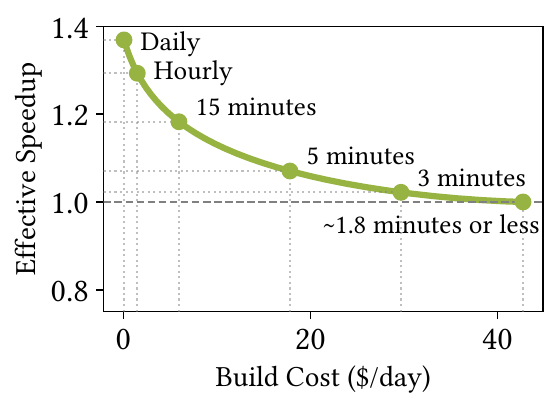}
    \caption{Redshift (\tpchp{})}
  \end{subfigure}
  \hfill
  \begin{subfigure}{0.23\textwidth}
    \includegraphics[width=\textwidth]{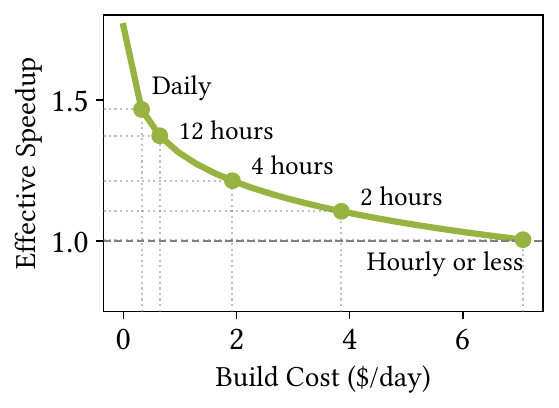}
    \caption{DuckDB (\tpchp{})}
  \end{subfigure}
  \caption{Impact of simulated bulk writes (e.g., ETLs) on \thesystem{}'s
    achieved speedup as we vary their frequency.}
  \label{fig:write-sweep}
\end{figure*}

\subsection{Query Run Time Prediction Robustness}
\label{sec:eval-robustness}
Now we study how the query run time predictor's accuracy affects \thesystem{}'s
decisions. We inject error into the ground truth query run times, and measure
the geomean speedup on our workloads. Figure~\ref{fig:qpp-robustness} shows our
results.
The horizontal axis is the injected Q-error; we test over-predictions up to a
Q-error of 13 and under-predictions (the shaded region) up to a Q-error of 1.99.

\textbf{\thesystem{} is robust to query run time prediction error on Redshift
and DuckDB up to a Q-error of 3.0.}
On Redshift, large query run time over-predictions cause \thesystem{} to use
accelerators whose data transfer time would exceed the the accelerator's benefit
because it thinks the original query is even slower. DuckDB is less sensitive to
query prediction errors because the data transfer time is negligible. However,
large over-predictions lead \thesystem{} to select slow accelerators that appear
faster than the inflated base query estimate, but are slower than the true time.

\textbf{\thesystem{} can achieve its results with existing query run time
predictors.}
On both Redshift and DuckDB, \thesystem{} maintains most of its performance up
to a prediction Q-error of 3.0 (1.6~on TPC-H+ for Redshift). Prior work on
query run time predictors shows that this accuracy is
achievable~\cite{stage-wu24, t3-rieger25, lce-sun19, qppnet-marcus19,
zeroshot-hilprecht22, unify-wu22, neo-marcus19, bao-marcus22, brad-yu24}.

\begin{table}[t]
  \footnotesize
  \caption{\thesystem{}'s query planning overhead.}
  \label{tbl:eval-overhead}
  \begin{tabularx}{\columnwidth}{@{}ll XXXX@{}}
    \toprule
    & & \multicolumn{2}{c}{\textbf{TPC-H}} &
      \multicolumn{2}{c}{\textbf{\tpchp{}}} \\
    \cmidrule(lr){3-4} \cmidrule(l){5-6}
    \textbf{RDBMS} & \textbf{Percentile} & \textbf{Time} & \textbf{Percent} &
      \textbf{Time} & \textbf{Percent} \\
    \midrule
    Redshift & p50 & 65.1~ms & 1.88\%   & 75.5~ms & 1.09\% \\
             & p90 & 109~ms & 43.7\%    & 202~ms & 3.14\% \\
    \midrule
    DuckDB   & p50 & 4.11~ms & 0.0900\% & 5.15~ms & 0.0804\% \\
             & p90 & 14.1~ms & 1.69\%   & 528~ms & 2.85\% \\
    \bottomrule
  \end{tabularx}
\end{table}

\subsection{Runtime Overhead}\label{sec:eval-overhead}
Table~\ref{tbl:eval-overhead} shows that \thesystem{}'s overhead, comprising
plan enumeration and performance prediction, is below 5\% in all but one case.
Relative overhead is lower on \tpchp{} than TPC-H because its queries run
longer. End-to-end latency still improves despite this overhead
(Figure~\ref{fig:tpch-e2e}).
For Redshift, most overhead comes from remote calls to the cardinality
estimator. Since our Redshift experiments ran outside AWS, they have higher
overhead than DuckDB. Replacing the remote call with an embedded estimator
would reduce latency.
For workloads comprising short-running queries, users can have \thesystem{} run
optimization in the background when it first encounters a query and use only
\emph{cached} accelerator execution plans.
Since industrial workloads have repetition~\cite{redset, redshift_pred_cache},
this approach would still allow speedups without slowing first-run latency.

\subsection{Impact of Writes}\label{sec:eval-writes}
\thesystem{}'s accelerators may use pre-computed state, which writes can make
stale. \thesystem{} targets engines running with bulk writes (common in OLAP
settings~\cite{olap_survey}).
After a bulk write, \thesystem{} routes queries to the underlying RDBMS instead
of its stale accelerators. In the background, it rebuilds the accelerators and
uses them until the next write. We leave incremental accelerator state
maintenance~\cite{ivm-zhou07, dbsp-budiu23} to future work.
Here, we study how write frequency affects \thesystem{}'s speedups. Under
this write model, \thesystem{} is beneficial when the time between writes
exceeds the rebuild time.
Using the 10\% budget results from \Cref{sec:eval-cs1,sec:eval-cs2}, we measure
speedups over a one day period by simulating write frequencies from once a day
to once a minute. Queries arriving during a rebuild run on the RDBMS; those
arriving after a rebuild but before a write use \thesystem{}.
We assume a cloud setting where rebuilds can be offloaded to additional compute
at additional cost.

Figure~\ref{fig:write-sweep} shows our results. The vertical axis is the
workload's geomean speedup over all queries. The horizontal axis is the build
cost, which is the time spent rebuilding the accelerators multiplied by the
on-demand EC2~\cite{ec2} price of an equivalent VM performing the rebuilding.
The labeled points represent different write frequencies (e.g., once every 5
minutes, etc.). When the write frequency equals or exceeds the build time, the
effective workload speedup is 1.0$\times$ because the rebuild will not keep up
with the new writes.

\textbf{\thesystem{} supports daily bulk writes with a negligible performance
impact.}
Our rebuild times are 47 seconds (Redshift, TPC-H), 25 minutes (DuckDB, TPC-H),
2 minutes (Redshift, \tpchp{}), and around 1 hour (DuckDB, \tpchp{}). Overall,
the build time is much less than 24 hours, meaning \thesystem{} can provide its
speedups with daily ETLs. Our Redshift setup has more resources available for
accelerator building, so it supports more frequent writes (e.g., every 15
minutes) with only a small speedup trade-off.

%% file: sections/06-related-work.tex
\section{Related Work}\label{sec:related-work}

\sparagraph{Database extensions, UDFs, UDOs.}
User-defined functions and operators (UDFs/UDOs)~\cite{tuplex, yesql,
udf_outline, udo-sichert}, extensions~\cite{pgextensions, duckdb-extensions},
and declarative sub-operators~\cite{dsos-jungmair23} let users add functionality
to an RDBMS.
However, these mechanisms are constrained by the host RDBMS's extension
framework, which (i) limits which accelerators it can support (e.g.,~Redshift
only supports scalar UDFs~\cite{redshift-udf-support}), (ii) may provide limited
optimizer support~\cite{graceful-udf-wehrstein25, qoff-chaudhuri93}, and (iii)
couples implementations to a particular RDBMS.
Declarative sub-operators~\cite{dsos-jungmair23}, in particular, require
modifying the underlying RDBMS to support their operator abstraction, whereas
\thesystem{} is non-invasive.

\sparagraph{Query plan patterns.}
Abstract operator trees~\cite{super-operators} and Calcite~\cite{calcite} also
provide methods to specify logical plan patterns, but both are less expressive
than \alps{}.
Abstract operator trees use abstract edges to capture a varying number of unary
operators~\cite{super-operators}, meaning they cannot match arbitrarily nested
binary operators (e.g., left-deep joins; see \Cref{apdx:aot} for details).
Calcite's rules define fixed-depth patterns, meaning they cannot match variable
length operator chains (e.g., arbitrarily nested filters) in a single rule
application.
\alps{} can express both patterns using repetitions.

\sparagraph{Auto materialized views.}
Choosing accelerator(s) to instantiate and use is similar to automated
materialized view (MV) selection~\cite{automv, mvsel-mqo, mvsel-mdim} and
exploitation~\cite{mvopt, mvopt-survey}.
Our problem setting differs in two key ways.
First, MV selection algorithms can materialize any sub-query.
\thesystem{} uses a fixed set of accelerator \alps{}, which constrains the
decision space.
Second, \thesystem{} makes query planning decisions over \alps{}, which each
represent a set of query fragments.
An MV is only a single concrete query sub-expression.

\sparagraph{Auto index selection.}
\thesystem{}'s accelerator instance selection problem is similar to automatic
index selection~\cite{automv, autoindex-survey, autoadmin-chaudhuri98,
msr_paper}.
The key difference is that accelerators are more general than indexes as they
can represent any query fragment, which requires a new approach.
Indeed, indexes are one specific kind of accelerator.

\sparagraph{Micro-adaptivity.}
Micro-adaptive techniques accelerate queries by finding low-level
optimizations that exploit the hardware and data properties during query
execution~\cite{microadaptivity-vectorwise, permutable-compiled-queries,
excalibur}.
Compared to \thesystem{}, micro-adaptive techniques occupy a different spot on
the speedup versus integration effort trade-off spectrum.
To get speedups using low-level optimizations and to rapidly switch between
optimizations at runtime, micro-adaptive techniques must be deeply integrated
into the query execution engine.
In contrast, \thesystem{} does not require code changes to the underlying RDBMS,
allowing it to even work with proprietary systems (e.g.,~Redshift).

\sparagraph{Federated databases.}
Federated databases~\cite{breitbart1992overview, breitbart1988update,
hwang1994myriad, pu1988superdatabases, sheth1990federated,
georgakopoulos1991multidatabase, josifovski2002garlic, bent2008dynamic,
zhang2022skeena, microsoft-fabric, insitu-crossdb-gavriilidis23} optimize
queries across multiple engines, focusing on splitting queries across
full-fledged databases unify access to multiple data sources.
In contrast, \thesystem{} focuses on intelligently offloading parts of a query onto
specialized query processors that cannot necessarily run any query fragment.

\sparagraph{Synthesizing query engines.}
Orthogonal work generates specialized query engines~\cite{castor-feser,
bespoke-olap, gendb, gpt-db}.
However, users must declare all query templates upfront and only queries
matching the template can then leverage the custom executor.
\thesystem{}, in contrast, replaces subqueries with accelerators. One
instantiated accelerator can thus support many queries (e.g., whenever appearing
as a subquery).
LLM-based query engine generators~\cite{bespoke-olap, gendb} could also help
generate \thesystem{} accelerators, which we leave to future work.

%% file: sections/07-conclusion.tex
\section{Conclusion}
We presented abstract logical plans (\alps{}) and \thesystem{}: a new
non-invasive query planning and execution framework that provides a practical
way to integrate query accelerators with any RDBMS that supports data
import/export.
\alps{} and a layer of indirection let \thesystem{} strike the sweet spot
between specializing the RBDMS for performance and maintaining its generality,
all without having to modify its source code.
Our three case studies show that users can integrate workload-specific
accelerators that \thesystem{} automatically applies when beneficial, achieving
geomean speedups of \BestGeomeanTpch{}, \BestGeomeanTpchp, and
\BestGeomeanSQLStorm{}.

%% file: sections/08-appendix.tex
\section*{Summary of Appendices}
Our appendices provide the following additional material for interested readers:
\begin{itemize}[leftmargin=*]
  \item Additional end-to-end experiments on TPC-H and \tpchp{} with a 1\% space
    budget (\Cref{apdx:tpch1,apdx:tpchp1}).
  \item An experiment studying the \alp{} neural network's ability to generalize
    to larger unseen dataset sizes (\Cref{apdx:model-gen}).
  \item Additional background and details on the NFTA e-graph matching algorithm
    that \thesystem{} uses (\Cref{apdx:matching}).
  \item A discussion on and a proof of \alp{}'s expressivity over abstract
    operator trees~\cite{super-operators} (\Cref{apdx:aot}).
\end{itemize}

\section{Additional Experiments}
\subsection{Case Study 1: TPC-H with 1\% Budget}\label{apdx:tpch1}
Figure~\ref{fig:tpch-e2e-apdx} shows our results from case study 1 (TPC-H) with
a 10\% and 1\% space budget. We draw the following conclusion.

\textbf{\thesystem{}'s selects good accelerator plans even with a constrained
space budget.}
On TPC-H, as we reduce the space budget from 10\% to 1\%, \thesystem{}
intelligently prunes its index accelerator instances for Q18 (Redshift and
DuckDB) and Q19 (DuckDB) \captionlabel{M} while preserving the CDF and domain
negation accelerators for Q1, Q4, and Q13 \captionlabel{N}. This result shows
how \thesystem{} prioritizes accelerators providing the most speedup within the
space budget.

\subsection{Case Study 2: \tpchp{} with 1\% Budget}\label{apdx:tpchp1}
Figure~\ref{fig:tpchplus-e2e-apdx} shows our results from case study 2
(\tpchp{}) with 10\% and 1\% space budgets. We draw the following conclusion.

\textbf{Like in our first case study, \thesystem{} selects good accelerator
plans with a constrained space budget.}
For Q9 on Redshift, both strategies perform the same with a 10\%
budget~\captionlabel{P}. However, they diverge at 1\%. The na\"ive strategy
chooses two domain negation accelerators that use less space compared to
\thesystem{}'s choice (to save space for other accelerators it greedily
selected), but this choice causes a slowdown. In contrast, \thesystem{}
correctly keeps its original accelerator choices to maintain the
speedup~\captionlabel{Q}.

\begin{figure*}[t]
  \centering
  \begin{subfigure}[t]{0.47\textwidth}
    \begin{overpic}[width=\textwidth]{figures/05a-redshift-e2e-10pct}
      \put(24,25){\captionlabelsmall{M}}
    \end{overpic}
    \caption{Redshift (10\%) (X indicates timeout)}
    \label{fig:tpch-redshift10-apdx}
  \end{subfigure}
  \qquad
  \begin{subfigure}[t]{0.47\textwidth}
    \begin{overpic}[width=\textwidth]{figures/05b-duckdb-e2e-10pct}
      \put(25,21.5){\captionlabelsmall{M}}
    \end{overpic}
    \caption{DuckDB (10\%)}
    \label{fig:tpch-duckdb10-apdx}
  \end{subfigure}

  \begin{subfigure}[t]{0.47\textwidth}
    \begin{overpic}[width=\textwidth]{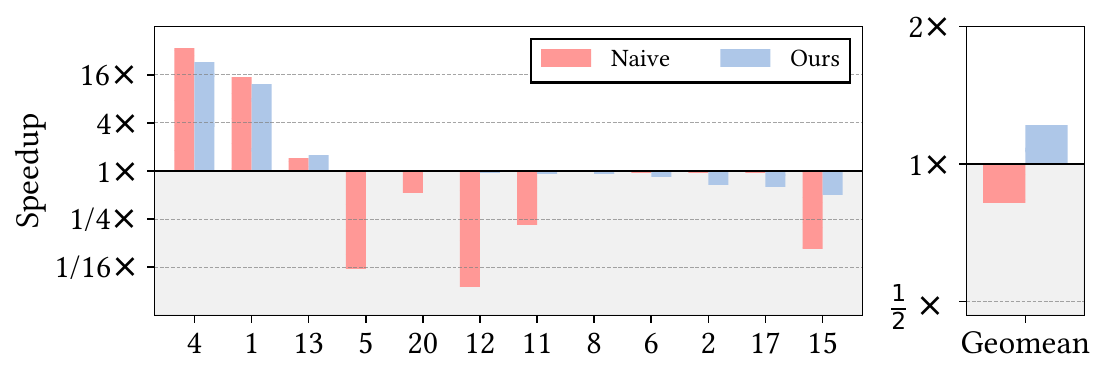}
      \put(21,29){\captionlabelsmall{N}}
    \end{overpic}
    \caption{Redshift (1\%)}
    \label{fig:tpch-redshift1-apdx}
  \end{subfigure}
  \qquad
  \begin{subfigure}[t]{0.47\textwidth}
    \begin{overpic}[width=\textwidth]{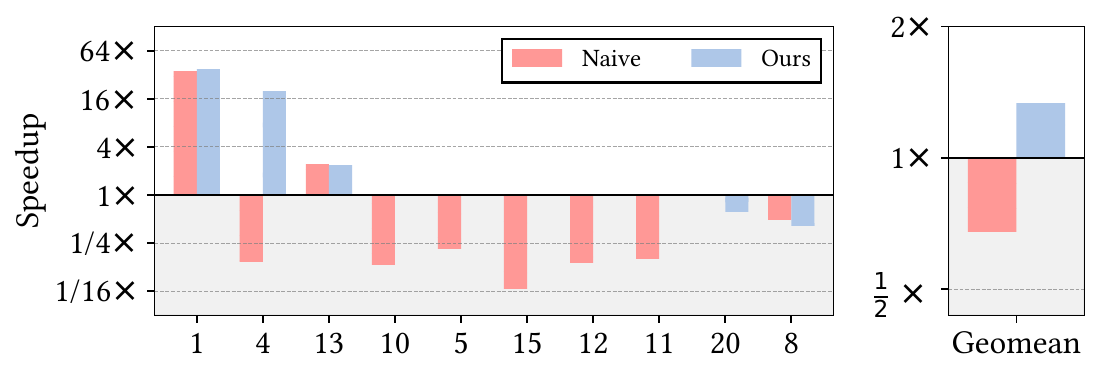}
      \put(23,28){\captionlabelsmall{N}}
    \end{overpic}
    \caption{DuckDB (1\%)}
    \label{fig:tpch-duckdb1-apdx}
  \end{subfigure}
  \caption{\thesystem{} accelerates TPC-H (SF = 100) on both Redshift and DuckDB
  by \OverallTpchRedshiftGeomeanSpeedup{} and \OverallTpchDuckDBGeomeanSpeedup{}
  respectively (geomean).}
  \label{fig:tpch-e2e-apdx}
\end{figure*}

\begin{figure*}
  \centering
  \begin{subfigure}[t]{0.47\textwidth}
    \begin{overpic}[width=\textwidth]{figures/06a-redshift-tpchp-10pct}
      \put(28,23){\captionlabelsmall{P}}
    \end{overpic}
    \caption{Redshift (10\%) (X indicates timeout)}
    \label{fig:tpchp-redshift10-apdx}
  \end{subfigure}
  \qquad
  \begin{subfigure}[t]{0.47\textwidth}
    \begin{overpic}[width=\textwidth]{figures/06b-duckdb-tpchp-10pct}
      \put(24.5,22){\captionlabelsmall{P}}
    \end{overpic}
    \caption{DuckDB (10\%)}
    \label{fig:tpchp-duckdb10-apdx}
  \end{subfigure}

  \begin{subfigure}[t]{0.47\textwidth}
    \begin{overpic}[width=\textwidth]{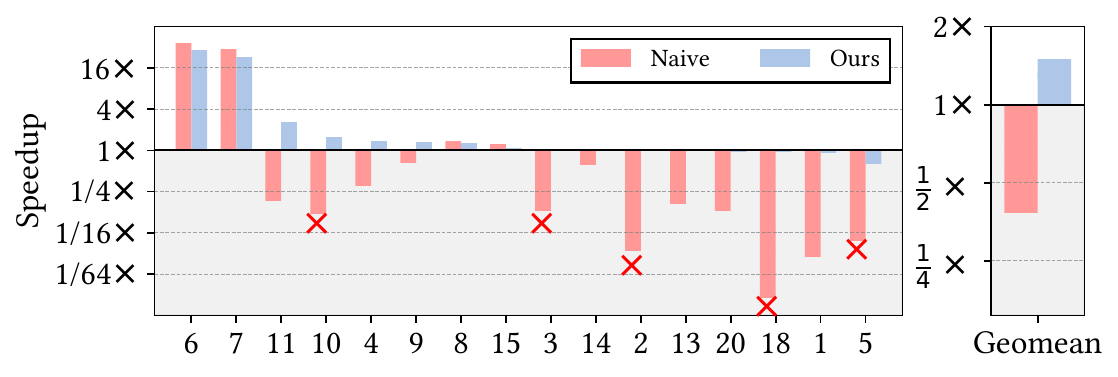}
      \put(35.5,23){\captionlabelsmall{Q}}
    \end{overpic}
    \caption{Redshift (1\%) (X indicates timeout)}
    \label{fig:tpchp-redshift1-apdx}
  \end{subfigure}
  \qquad
  \begin{subfigure}[t]{0.47\textwidth}
    \begin{overpic}[width=\textwidth]{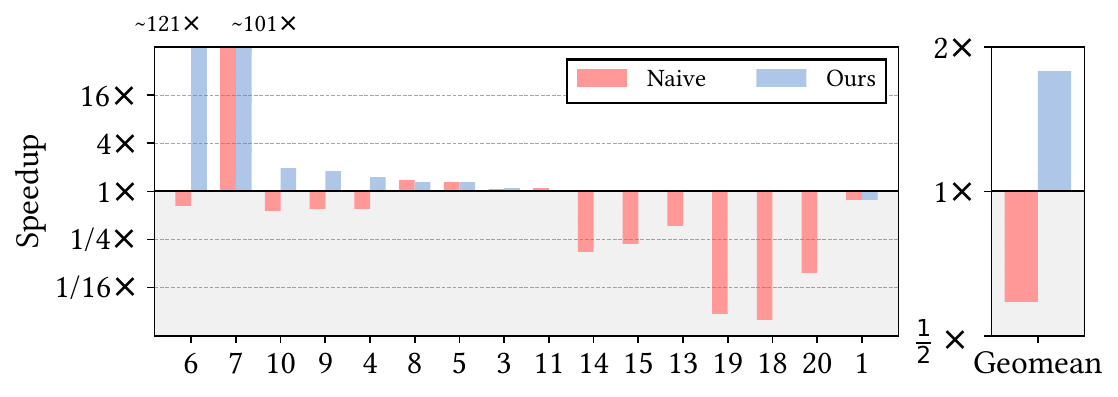}
      \put(27,23){\captionlabelsmall{Q}}
    \end{overpic}
    \caption{DuckDB (1\%)}
    \label{fig:tpchp-duckdb1-apdx}
  \end{subfigure}
  \caption{\thesystem{} accelerates \tpchp{} on Redshift and DuckDB
  by \OverallTpchpRedshiftGeomeanSpeedup{} and \OverallTpchpDuckDBGeomeanSpeedup{}
  respectively (geomean).}
  \label{fig:tpchplus-e2e-apdx}
\end{figure*}

\subsection{Model Generalization to Larger Datasets}\label{apdx:model-gen}
We evaluate the out-of-distribution generalization of the \alp{} neural
network as the dataset size increases. Specifically, we train the network for
domain negation on Redshift using TPC-H scale factors (SF) 1 and 10, and test on
SF = 100. We compare this generalizing configuration (denoted by the ``gen.''
suffix) against the same model trained on all three scale factors. We repeat
this setup for a standard MLP too. Figure~\ref{fig:model-gen} shows the p50 test
Q-error, where lower is better. \textbf{From these results, we can conclude that
that the ALP network generalizes better than the MLP to unseen dataset sizes.}
The \alp{} network's p50 Q-error increases by only 0.56 when generalizing,
compared to a larger increase of 1.00 for the MLP.

\begin{figure}
  \centering
  \includegraphics[width=\columnwidth]{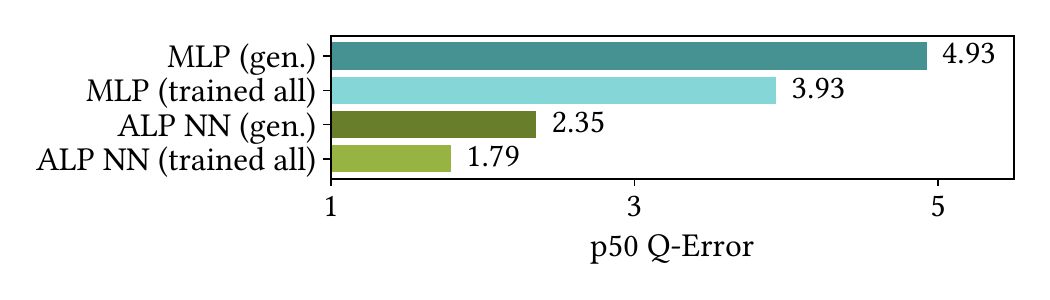}
  \caption{Model generalization across dataset sizes.}
  \label{fig:model-gen}
\end{figure}

\section{NFTA Matching on E-Graphs}\label{apdx:matching}

\sparagraph{Background on NFTAs.}
To help with understanding our matching algorithm, we first give some background
on NFTAs. An NFTA is a state machine for trees. Formally, it is a 4-tuple $(Q,
F, Q_f, \Delta)$ comprising a set of states $Q$, accepting states $Q_f$ where
$Q_f \subseteq Q$, an alphabet $F$, and a set of transitions
$\Delta$~\cite{tree-acceptors-doner70, tree-automata-thatcher68, tata}.
Informally, $F$ is the set of nodes that can appear in an expression tree. We
use the notation $a(c_1, c_2, \dots, c_k)$ to refer to a node $a$ that has $k$
child nodes $c_1$ to $c_k$. For example, a filter in a logical plan could be
represented as $\text{Filter}(x, p)$ where $x$ is its child plan operator and
$p$ represents its predicate expression. Transitions are given as $(l_1, \dots
l_k) \rightarrow_{t} r$ where $l_1, \dots, l_k, r \in Q$. If we encounter a node
$t$ whose $k$ children are in states $l_1$ to $l_k$, we transition to state $r$.
A transition for a leaf node $t$ is expressed as $() \rightarrow_{t} r$.
A bottom-up NFTA matches an expression tree if there exists a sequence of states
from its leaf nodes to its root and the root state is in $Q_f$~\cite{tata}.

\sparagraph{Background on E-Graphs.}
E-graphs are data structures that represent multiple equivalent tree expressions
in a single data structure~\cite{egraphs-thesis-nelson, egraphs-book}, much like
logical groups in a Cascades optimizer~\cite{cascades-graefe95}.
Concretely, they are directed graphs of e-classes, which contain e-nodes. An
e-node is any node that would appear in a concrete tree with the twist that its
children are e-classes. The e-nodes in the same e-class describe logically
equivalent expressions. For example, $\text{Filter}(\text{Filter}(x, p_1), p_2)$
is equivalent to $\text{Filter}(x, p_1 \wedge p_2)$, so their top-level nodes
would be in the same e-class. Critically, e-graphs are unlike expression trees
because they can contain cycles.

\begin{algorithm}[t]
  \caption{Top-down NFTA matching on an e-graph.\\Our two extensions are
  \textcolor{ForestGreen}{highlighted in green.}}
  \small
  \label{alg:nfta-match}
  \begin{algorithmic}
    \Require $N$: NFTA, $E$: Query e-graph, $R$: Root e-class
    \Ensure True if $N$ matches $E$ at $R$, False otherwise
    \LComment{$S$: Curr. states, $\Delta$: Transitions, $C$: Current e-class, $P$: Match path}
    \Function{TopDownMatch}{$S, \Delta, C, P$}
      \LComment{Traverse the transitions and find child states.}
      \State $M \gets $ [ ]
      \ForAll{transition $((l_1, l_2, \dots, l_k) \rightarrow_{t} r) \in \Delta$ and $r \in S$}
        \If{{\color{ForestGreen} $t$ in e-class $C$ and it has $k$ children}}
          \If{{\color{ForestGreen}$(r, C)$ does not create a cycle in $P$}}
            \State Append $((l_1, l_2, \dots, l_k) \rightarrow_{t} r)$ to $M$
          \EndIf{}
        \EndIf{}
      \EndFor{}
      \If{any transition in $M$ is to a leaf (e.g., $() \rightarrow_{t} r$)}
        \State \Output True
      \EndIf{}
      \LComment{Recurse to children.}
      \State $M' \gets $ {\color{ForestGreen}Group transitions in $M$ by
        terminal $t$. Collect child states by position into sets.}
      \ForAll{{\color{ForestGreen} $((L_1, L_2, \dots L_k) \rightarrow_{t} R )$ in $M'$}}
        \State $Out \gets $ [ ]
        \ForAll{$i$ from $1$ to $k$}
          \State $C_i \gets $ terminal $t$'s child $i$
          \State $P' \gets $ {\color{ForestGreen} $P \cup \{(r, C) \mid r \in R\}$}
          \State Append \Call{TopDownMatch}{$L_i, \Delta, C_i, P'$} to $Out$
        \EndFor{}
        \If{$Out$ is all True}
          \State \Output True
        \EndIf{}
      \EndFor{}
      \State \Output False
    \EndFunction{}
    \State
    \State $(Q, Q_F, \Delta) = N$ \Comment{$Q$: States, $Q_F$: Final states, $\Delta$: Transitions}
    \State \Output \Call{TopDownMatch}{$Q_F, \Delta, R, \{\}$}
  \end{algorithmic}
\end{algorithm}

\sparagraph{NFTA Matching Algorithm on E-Graphs.}
Algorithm~\ref{alg:nfta-match} shows our matching algorithm, which makes two
extensions ({\color{ForestGreen} highlighted in green}) to NFTA top-down
matching~\cite{tata}. \thesystem{} runs this matching algorithm on each e-class
in the e-graph.

Our first extension is to search over all e-nodes in the e-class. We accept
(match) the root e-class if there is at least one e-node in the root class where
all of its children reach a leaf transition. We examine all e-nodes in the
e-class because the NFTA might only match a tree rooted at one of the nodes in
the class (i.e., corresponding to a specific way of expressing the accelerator's
plan fragment).

Our second extension is in rejecting match cycles. NFTA transitions can contain
state cycles and e-graphs can contain cycles. Existing match algorithms do not
have this problem because tree expressions do not have cycles, so it is
impossible to get stuck in a match cycle. Our key observation is to reject
matching paths that contain \emph{both} a transition cycle and e-class cycle.
Rejecting only transition cycles is incorrect because transition cycles describe
repetitions in the tree expression. Rejecting only e-class cycles is incorrect
because the NFTA could be matching a tree expression that has a finite number of
recursive repetitions.
We track the \emph{path} of NFTA states and e-classes during a match and stop
following any transitions that would create both a state and e-class cycle.

\section{\alps{} Compared to Abstract Operator Trees}\label{apdx:aot}

As discussed in Section~\ref{sec:related-work}, abstract operator
trees~\cite{super-operators}, which describe sets of super operators, are
similar in spirit to \alps{} as they both ``describe query plan patterns.''
However, abstract operator trees are different from \alps{}, in part, because
they are strictly less expressive than \alps{}. We prove this statement below.

\begin{theorem}
Abstract operator trees~\cite{super-operators} are strictly less expressive than
\alps{}.
\end{theorem}
\begin{proof}
We prove the theorem by (i) showing that there exists a set of query plans that
\alp{} templates can recognize (i.e., match) but abstract operator trees cannot,
and (ii) that every abstract operator tree can be described by an \alp{}
template.

\sparagraph{Part (i).}
Suppose we want to match a left-deep inner join over base tables of arbitrary
depth. Let $\Bowtie(l, r)$ represent an inner join over $l$ and $r$. We omit the
join conditions for clarity. We can express this pattern as an \alp{} template
$T$ using a repetition:
\begin{align*}
    T &= \mathbf{Rep}(\mathrm{repeated\text{-}join}; T_B, T_R, c) \\
    T_B &= \mathbf{Var}(\mathrm{left}; \mathrm{TableRef}) \\
    T_R &= \Bowtie(c, \mathbf{Var}(\mathrm{right}; \mathrm{TableRef})
\end{align*}
This template captures arbitrarily left-deep nested $\Bowtie$ operations as the
repetition is rooted at the $c$ token in $T_R$.

Abstract operator trees use abstract edges to represent operators that can
repeat zero or more times. However, abstract edges only support \emph{unary
operators} (e.g., selections). Inner joins are binary operators. Thus this
pattern cannot be expressed using an abstract operator tree.

\sparagraph{Part (ii).}
An abstract operator tree comprises core operators connected with abstract
edges~\cite{super-operators}. We can construct an \alp{} template from an
abstract operator tree by converting each core operator into a token and each
abstract edge into a repetition over an alternation of the abstract edge's
operator options. Stated precisely, let $C$ be a core operator. Then a template
$T_C$ for the core operator is
\begin{align*}
  T_C &= C(T_I)
\end{align*}
where $T_I$ is a template for the incoming abstract edge or other core operator.
The spool operator (base of the abstract operator tree) has no incoming edge, so
it is represented using just a leaf token $T_{\mathrm{Spool}} = \mathrm{Spool}$.
Let $e$ be an abstract edge. A template $T_e$ for the abstract
edge is
\begin{align*}
  T_e &= \mathbf{Rep}(\mathrm{abs\text{-}edge}; T_I, T_R, c) \\
  T_R &= \mathbf{Alt}(\mathrm{abs\text{-}edge\text{-}options}; e_1(c), e_2(c), \dots, e_n(c))
\end{align*}
where $T_I$ represents a template for the abstract edge's incoming operator and
$e_1, \dots, e_n$ represent the unary operator options for the abstract edge
(e.g., selection $\sigma$, projection $\pi$, etc.).

\sparagraph{Conclusion.}
By the arguments in parts (i) and (ii), any abstract operator tree can be
expressed as an \alp{} and there exists a pattern that is expressible by \alps{}
but not abstract operator trees. Thus abstract operator trees are strictly less
expressive than \alps{}.
\end{proof}